\documentclass[12pt,english]{article}
\usepackage{geometry}
\usepackage{textcomp}
\usepackage{amsmath}
\usepackage{amsthm}
\usepackage{amssymb}
\usepackage{mathtools}
\usepackage{setspace}
\usepackage{babel}
\usepackage{color}
\usepackage{hyperref}
\usepackage{tabularx}
\usepackage{booktabs}
\usepackage{placeins}
\usepackage{comment}
\usepackage{subcaption}
\usepackage{csquotes}

\usepackage{pdflscape}
\usepackage{float}
\usepackage{array}

\theoremstyle{plain}
\newtheorem{thm}{Theorem}
\theoremstyle{remark}

\theoremstyle{definition}
\newtheorem{assumption}{Assumption}

\usepackage[
  backend=biber,
  style=apa,
  sorting=ynt,
  doi=false, 
  url=false, 
  isbn=false
]{biblatex} 
\begin{document}
\title{Dyadic data with ordered outcome variables}
\author{
  Chris Muris\thanks{Department of Economics, McMaster University, email: muerisc@mcmaster.ca.} \and 
  Cavit Pakel\thanks{Department of Economics, University of Oxford, email: cavit.pakel@economics.ox.ac.uk.} \and 
  Qichen Zhang\thanks{Department of Economics, McMaster University, email: zhanq73@mcmaster.ca.}
}
\date{First version: February, 2023. This version: \today.}
\maketitle

\begin{abstract}
We consider ordered logit models for directed network data that allow for flexible sender and receiver fixed effects that can vary arbitrarily across outcome categories.
This structure poses a significant incidental parameter problem, particularly challenging under network sparsity or when some outcome categories are rare.
We develop the first estimation method for this setting by extending tetrad-differencing conditional maximum likelihood (CML) techniques from binary choice network models.
This approach yields conditional probabilities free of the fixed effects, enabling consistent estimation even under sparsity.
Applying the CML principle to ordered data yields multiple likelihood contributions corresponding to different outcome thresholds.
We propose and analyze two distinct estimators based on aggregating these contributions: an Equally-Weighted Tetrad Logit Estimator (ETLE) and a Pooled Tetrad Logit Estimator (PTLE).
We prove PTLE is consistent under weaker identification conditions, requiring only sufficient information when pooling across categories, rather than sufficient information in each category.
Monte Carlo simulations confirm the theoretical preference for PTLE, and an empirical application to friendship networks among Dutch university students demonstrates the method's value.
Our approach reveals significant positive homophily effects for gender, smoking behavior, and academic program similarities, while standard methods without fixed effects produce counterintuitive results.
\end{abstract}

\section{Introduction}
\label{sec:introduction}

Dyadic data, capturing interactions between pairs of agents such as individuals, firms, or countries, are common in economics and social sciences.
While many studies focus on the binary existence of a link, these relationships commonly have ordered levels.
For example, countries may choose different alliance strengths, ranging from non-alignment to defense treaties \parencite{laiDemocracyPoliticalSimilarity2000a}.
Similarly, individuals rate their social relationships on scales from ``troubled'' to ``best friendship,'' as in the student network data we analyze later \parencite{vanduijnFrameworkComparisonMaximum2009}.
To estimate these relationship determinants reliably, we need to account for the agents' overall tendency to connect with others.

This paper develops methods for analyzing ordered outcomes in directed dyadic data while controlling for agent-specific heterogeneity.
We specify an ordered logit model where the outcome $Y_{ij} \in \{0, 1, \dots, M\}$ for each dyad $(i,j)$ depends on covariates $X_{ij}$ and node-specific fixed effects.
Our specification includes a sender effect $\lambda_{im}$ and a receiver effect $\delta_{jm}$ that can vary across outcome categories, allowing for flexible heterogeneity patterns.
This specification presents estimation challenges due to the large number of fixed effects, especially in sparse networks where certain outcome categories are rare.

Our primary contribution is developing a practical way to estimate ordered logit models with flexible category-specific fixed effects for network data.
We adapt the conditional likelihood approach based on tetrad differencing, which has worked well for binary network models.
This technique compares specific patterns within groups of four nodes (tetrads) to eliminate the fixed effects.
By constructing these comparisons with binarized outcomes, we obtain probabilities that depend only on the covariates and the parameters of interest.
This approach does not require directly estimating the fixed effects and works well even in sparse networks, where many potential relationships are not observed or certain outcome categories are rare.

Our second key contribution concerns how to best use information from different outcome thresholds.
When we transform the ordered outcome into binary indicators (i.e., whether $Y_{ij} \geq m$), we get $M$ different sets of conditional probabilities, one for each threshold.
This structure suggests two natural estimation approaches.
The first estimator (ETLE) gives equal weight to the information from each threshold.
The second estimator (PTLE) pools all informative comparisons into a single estimation.
We show that PTLE is preferable because it works under weaker conditions.
While ETLE requires enough information at each individual threshold, PTLE only needs enough information when combining all thresholds together.
This makes PTLE more reliable when certain outcome categories are rare, which is common in network data.

Third, we explore how the identification approach changes with different fixed effects structures.
Besides our main flexible model, we analyze two simpler versions: one where node effects ($\alpha_i, \gamma_j$) are constant across categories with common thresholds ($\lambda_m$), and another with category-specific sender effects but category-invariant receiver effects ($\lambda_{im}, \gamma_j$).
We show that the tetrad comparisons must be constructed differently for each model structure to properly eliminate the fixed effects.
This reveals how model assumptions affect the identification strategy.

Our work connects to several related bodies of literature while making distinct contributions.
In comparison to the panel data ordered choice literature \parencite[\textit{e.g.,}][]{dasPanelDataModel1999,johnsonPanelDataModels2004,Baetschmann2015,murisEstimationFixedEffectsOrdered2017, botosaruIdentificationTimevaryingTransformation2023}, we tackle the distinct challenge of network dependence from overlapping dyads.
Relative to the network formation literature \parencite{charbonneauMultipleFixedEffects2017,grahamEconometricModelNetwork2017, jochmansSemiparametricAnalysisNetwork2018}, our primary innovation lies in addressing ordered outcomes.
This introduces complexities related to multiple thresholds, category-specific sparsity, and the consequent PTLE versus ETLE estimator choice, aspects not present in binary models.
Our use of conditional likelihood also distinguishes our approach from methods employing bias correction for fixed effects in networks \parencite{dzemskiEmpiricalModelDyadic2019,hughesEstimatingNonlinearNetwork2023a}.
Concurrent work by \textcite{sziniPairwiseDifferencingDistribution2025} also employs tetrad differencing for network fixed effects models but focuses on estimating distribution regression coefficients $\beta(y)$ for potentially continuous outcomes, rather than the single common parameter $\beta$ in an ordered logit specification.

We validate our theoretical results through extensive Monte Carlo simulations.
These experiments confirm the consistency of the PTLE under weaker conditions and highlight its finite-sample performance compared to the ETLE, particularly when identification is weak for some categories or the network is sparse.
Furthermore, we illustrate the practical application of the PTLE method by analyzing the determinants of friendship intensity within a network of Dutch university students \parencite{vanduijnFrameworkComparisonMaximum2009}, showing that flexible node-specific heterogeneity is essential for reliable inference: standard approaches without fixed effects yield negative coefficient signs where positive homophily is expected, while our method uncovers significant positive homophily effects across gender, smoking behavior, and academic program dimensions.

The rest of the paper proceeds as follows.
Section \ref{sec:modelII} details the model specification.
Section \ref{sec:identification} presents the identification strategy for the main model.
Section \ref{sec:estimationII} develops the ETLE and PTLE estimators and provides consistency results.
Section \ref{sec:simulation} reports the simulation findings.
Section \ref{sec:empirical-application} discusses the empirical application.
Section \ref{sec:alternative_specifications} analyzes the alternative specifications.
The appendix contains all proofs.
\section{Model}
\label{sec:modelII}

Consider a set of $N$ nodes $\mathbb{N}_{N} = \left\{ 1,\ldots,N\right\}$.
Denote the set of all pairs by
$$\mathcal I_N = \left\{(i,j):\;
i \in \mathbb{N}_{N},\;
                 j \in \mathbb{N}_{N},\; 
                 i\neq j\right\}.$$ 
For each such dyad,%
\footnote{
    Our results can be generalized to settings where data are available only on a subset of dyads.
    The assumption that data are available for each dyad simplifies notation.
}
we observe an ordered dependent variable $Y_{ij}\in\left\{ 0,\cdots,M\right\}$ and a vector of regressors $X_{ij}\in\mathbb{R}^{k}$.

In our model, the observed outcome $Y_{ij}$ is determined by an underlying latent variable $Y^*_{ij}$ crossing category-specific thresholds $\lambda_{ijm}^*$:
\begin{equation}
Y_{ij} \geq m \Leftrightarrow Y^*_{ij} \geq \lambda_{ijm}^{*}, \quad m=1,\cdots,M. \label{eq:threshold}
\end{equation}
The latent variable represents the unobserved propensity for a relationship between nodes $i$ and $j$ and is defined as $Y^*_{ij}=X_{ij}^\prime \beta-\epsilon_{ij}$.
We use the convention $\lambda_{ij0}^* = -\infty$, consistent with $Y_{ij} \ge 0$ always holding, and require thresholds to be ordered: $\lambda_{ijm'}^* \geq \lambda_{ijm}^*$ whenever $m' \geq m$.
This formulation extends the standard ordered choice model to dyadic data.

We assume that the error terms $\epsilon_{ij}$ follow the standard logistic distribution with cumulative distribution function $\Lambda(z) = e^z / (1+e^z)$.
We further assume that $\epsilon_{ij}$ are independent across dyads and independent of regressors $\mathbf X = (X_{ij} : (i,j) \in \mathcal I_N)$.

If $M = 1$, the specification reduces to a binary choice model with $Y_{ij} \in \{0,1\}$.
This case is equivalent to the network formation models studied in \textcite{charbonneauMultipleFixedEffects2017} and \textcite{jochmansSemiparametricAnalysisNetwork2018}, see also \textcite{grahamEconometricModelNetwork2017}.

In our main specification, we impose structure on the thresholds $\lambda_{ijm}^*$.
We decompose it into sender-specific and receiver-specific components that can vary across outcome categories $m$:
\[
\lambda_{ijm}^{*} = \lambda_{im} + \delta_{jm}.
\]

This decomposition has a natural interpretation in many applications. 
For instance, in a friendship formation context like our empirical application, individuals rate relationships on a six-point scale.
Some individuals might consistently require higher thresholds to categorize relationships at higher levels (captured by variation in $\lambda_{im}$ across senders $i$), while some individuals might be more easily considered favorably by others (captured by variation in $\delta_{jm}$ across receivers $j$).
The specification permits these threshold effects to vary differently across categories for each agent.
For example, some students might exhibit steep threshold increases between lower and higher relationship categories, while others might progress more gradually across the full range of relationship types.
This heterogeneity in threshold profiles is fully accommodated by allowing both $\lambda_{im}$ and $\delta_{jm}$ to vary freely across both nodes and categories.

Throughout, we adopt a fixed effects approach, in the sense that we do not impose restrictions on the relationship between the threshold components $(\lambda_{im}, \delta_{jm})$ and the covariates $X_{ij}$.
These components are treated as incidental parameters.

The data available to the researcher are the covariates $\mathbf X$ and the outcomes $\mathbf Y = (Y_{ij} : (i,j) \in \mathcal I_N)$.
Let $\beta_0$ denote the true value of the regression coefficient vector $\beta$.
Let $F_i = ((\lambda_{im}, \delta_{im}): m = 1,\cdots,M)$ represent the collection of threshold parameters associated with node $i$ as both sender and receiver.
We gather all these fixed effects across the network into the vector $\mathbf F = (F_i : i \in \mathbb N_N)$, containing a total of $2NM$ incidental parameters.

Assumption \ref{a:likelihood_II} formally states the conditional likelihood based on these components.
\begin{assumption}\label{a:likelihood_II}
    The likelihood of the ordered choices $\mathbf Y = \mathbf y$ conditional on the covariates $\mathbf X$ and the parameters $(\beta_0, \mathbf{F})$ is
    \[
        P(\mathbf Y = \mathbf y| \mathbf X, \mathbf{F}) = \prod_{(i,j)\in\mathcal I_N} P(Y_{ij} = y_{ij}|X_{ij}, F_i, F_j),
    \]
    with the likelihood for a single dyad given by
    \[
    P(Y_{ij} = m | X_{ij}, F_i, F_j) =
    \begin{dcases}
    1 - \Lambda\left( X_{ij}^\prime \beta_0 - \lambda_{i1} - \delta_{j1}\right), & \text{if } m = 0, \\
    \begin{aligned}
        &\Lambda\left( X_{ij}^\prime \beta_0 - \lambda_{im} - \delta_{jm} \right) \\
        &- \Lambda\left( X_{ij}^\prime \beta_0 - \lambda_{i,m+1} - \delta_{j,m+1} \right),
    \end{aligned}
    & \text{if } 1 \leq m \leq M-1, \\
    \Lambda\left( X_{ij}^\prime \beta_0 - \lambda_{iM} - \delta_{jM} \right), & \text{if } m = M.
    \end{dcases}
    \]
\end{assumption}

\section{Identification}
\label{sec:identification}

In this section, we establish the identification of the regression coefficient $\beta_0$.
A key challenge in estimation is that the number of incidental parameters $\mathbf F$ grows with the sample size, leading to an incidental parameter problem.
We address this by developing an approach that eliminates these fixed effects from the likelihood function.

We begin by transforming our ordered outcome variable $Y_{ij}$ into a series of binary variables based on exceeding specific thresholds.
For each possible threshold value $m \in \{1,\cdots,M\}$, we define:
\[
D_{ij}(m) = \mathbf{1}\{Y_{ij} \geq m\}, \quad \forall (i,j)\in\mathcal I_N.
\]
This transformation is standard ordered choice models, allowing us to reframe the problem in terms of binary outcomes.
Under Assumption \ref{a:likelihood_II},
\begin{align}
P(D_{ij}(m) = 1 | X_{ij}, F_i, F_j) 
&= P(Y_{ij} \geq m | X_{ij}, F_i, F_j) \\ 
&= \Lambda( X_{ij}^\prime \beta_0 - \lambda_{im} - \delta_{jm}).
\end{align}
This probability depends on the node-specific fixed effects $\lambda_{im}$ and $\delta_{jm}$.

Our identification strategy builds on advances in binary choice models for dyadic data by using tetrads.
In our context, a tetrad is a set of four distinct nodes comprising two sending nodes ($i_1, i_2$) and two receiving nodes ($j_1, j_2$).
We define the set of all $q_N = \binom{N}{2} \binom{N-2}{2}$ possible tetrads as:
\begin{equation}
    \Sigma = \{\sigma = (i_1,i_2,j_1,j_2) \in \mathbb N_N^4 : 
               i_1, i_2, j_1, j_2 \text{ are all distinct nodes}
             \}.
\label{def:Sigma}
\end{equation}
A generic tetrad $\sigma$ can exhibit none, some, or all of the potential dyadic relationships: $(i_1,j_1)$, $(i_1, j_2)$, $(i_2, j_1)$, and $(i_2,j_2)$.

To identify the parameters, we analyze the pattern of binary outcomes within each tetrad.
We allow for different cutoffs $m_{rs}$ for each dyad $(i_r, j_s)$ within the tetrad.
Let $\mathbf m = (m_{11},m_{12},m_{21},m_{22})$ denote a vector of four potentially different cutoffs.

For any tetrad $\sigma$ and cutoff vector $\mathbf{m}$, we define a difference-in-differences statistic based on the binary variables:
\begin{equation}
    Z_\sigma(\mathbf m) = \frac{1}{2}\left(
        (D_{i_1,j_1}(m_{11}) - D_{i_1,j_2}(m_{12})) 
                             - 
        (D_{i_2,j_1}(m_{21}) - D_{i_2,j_2}(m_{22}))
    \right).
\label{def:Z}
\end{equation}
This construction generalizes the approach of \textcite{charbonneauMultipleFixedEffects2017} and \textcite{jochmansSemiparametricAnalysisNetwork2018} for binary outcomes.

Identification relies on cases where $Z_\sigma(\mathbf m) \in \{-1,1\}$, representing specific contrasting patterns of outcomes within the tetrad:
\begin{itemize}
    \item $Z_\sigma(\mathbf m) = +1 \Leftrightarrow D_{i_1,j_1}(m_{11}) = 1, D_{i_1,j_2}(m_{12}) = 0, D_{i_2,j_1}(m_{21}) = 0, D_{i_2,j_2}(m_{22}) = 1$
    \item $Z_\sigma(\mathbf m) = -1 \Leftrightarrow D_{i_1,j_1}(m_{11}) = 0, D_{i_1,j_2}(m_{12}) = 1, D_{i_2,j_1}(m_{21}) = 1, D_{i_2,j_2}(m_{22}) = 0$.
\end{itemize}
Below we will show that these configurations allow for conditional inferences free of fixed effects.

Let $X_\sigma = (X_{i_1,j_1},X_{i_1,j_2},X_{i_2,j_1},X_{i_2,j_2})$ collect the covariates for the tetrad.
Define the differenced covariate vector as:
\begin{align*}
    r_\sigma &= (X_{i_1,j_1} - X_{i_1,j_2}) - (X_{i_2,j_1} - X_{i_2,j_2}).
\end{align*}

The conditional probability of $Z_\sigma(\mathbf m)=1$ given $Z_\sigma(\mathbf m) \in \{-1, 1\}$ depends on $r_\sigma'\beta_0$ and a difference-in-differences term involving the fixed effects:
\begin{align}
    \Delta \lambda(\mathbf{m}) &=
    (\lambda^*_{i_1,j_1,m_{11}} - \lambda^*_{i_1,j_2,m_{12}}) - (\lambda^*_{i_2,j_1,m_{21}} - \lambda^*_{i_2,j_2,m_{22}}) \\
    &=  (
            (\lambda_{i_1,m_{11}} + \delta_{j_1,m_{11}})
            - 
            (\lambda_{i_1,m_{12}} + \delta_{j_2,m_{12}})
        )
        \nonumber \\
        &\quad - 
        (
            (\lambda_{i_2,m_{21}} + \delta_{j_1,m_{21}})
            - 
            (\lambda_{i_2,m_{22}} + \delta_{j_2,m_{22}})
        ).
\label{def:DeltaLambda}
\end{align}

As shown in the proof of Theorem \ref{thm:directed_II_sufficiency}, Assumption \ref{a:likelihood_II} implies:
$$
    P\left(\left. Z_{\sigma}(\mathbf m) = 1 \right| Z_{\sigma}(\mathbf m) \in \{-1,+1\}, X_\sigma, \mathbf{F} \right) 
    = 
    \Lambda\left(
    r_\sigma^\prime \beta_0 - 
    \Delta \lambda(\mathbf{m}) \right).
$$

For this conditional probability to yield inference on $\beta_0$ free of the fixed effects $\mathbf{F}$, the term $\Delta \lambda(\mathbf{m})$ must vanish.
Examining the structure of $\Delta \lambda(\mathbf{m})$ in Equation \eqref{def:DeltaLambda}, we see that terms like $\lambda_{i_1,m_{11}}$ and $\lambda_{i_1,m_{12}}$ only cancel if $m_{11} = m_{12}$.
Similarly, cancellation of the $\lambda_{i_2,\cdot}$ terms requires $m_{21} = m_{22}$, cancellation of the $\delta_{j_1,\cdot}$ terms requires $m_{11} = m_{21}$, and cancellation of the $\delta_{j_2,\cdot}$ terms requires $m_{12} = m_{22}$.
Therefore, for $\Delta \lambda(\mathbf{m})$ to be zero solely based on the structure (without imposing restrictions on the values of the fixed effects), we must impose the constraint that all four cutoffs are equal: $m_{11} = m_{12} = m_{21} = m_{22} = m$.%
\footnote{
    In Section \ref{sec:alternative_specifications}, we consider identification under additional structural restrictions on the thresholds. These more restrictive specifications allow for a richer variety of cutoff combinations $\mathbf{m}$ and enable identification of threshold differences, as shown in related work such as \textcite{baetschmannIdentificationEstimationThresholds2012, murisEstimationFixedEffectsOrdered2017}.}

For any single cutoff $m \in \{1,\cdots,M\}$, define:
\begin{align}
    \overline Z_\sigma(m) &\equiv Z_\sigma((m,m,m,m)) \\
                          &= \frac{1}{2}\left(
                                (D_{i_1,j_1}(m) - D_{i_1,j_2}(m)) 
                                      - 
                                (D_{i_2,j_1}(m) - D_{i_2,j_2}(m))
                            \right). \label{def:Z_bar}
\end{align}
This notation $\overline Z_\sigma(m)$ represents the tetrad statistic when the same cutoff $m$ is applied to all four dyads.
We call tetrads with $\overline{Z}_\sigma(m) \in \{-1,+1\}$ informative tetrads.

We can now state our main identification result:

\begin{thm}[Sufficiency]\label{thm:directed_II_sufficiency}
    Under Assumption \ref{a:likelihood_II}, for any tetrad $\sigma$ and any cutoff $m \in \{1,\cdots,M\}$, the conditional probability of $\overline{Z}_\sigma(m)=1$ given that the tetrad is informative ($\overline{Z}_\sigma(m) \in \{-1, +1\}$) and given the covariates $X_\sigma$ is:
    $$
    P\left(\left. \overline Z_{\sigma}(m) = 1 \right| \overline Z_{\sigma}(m) \in \{-1,+1\}, X_{\sigma}\right) = \Lambda\left(r_\sigma^\prime \beta_0 \right).
    $$
\end{thm}
\begin{proof}
See Appendix \ref{sec:proofs}, page \pageref{proof:sufficiency_II}. 
\end{proof}

This theorem demonstrates that by conditioning on informative tetrads, we obtain a conditional probability for the binary outcome $\overline Z_\sigma(m)=1$ that follows a standard logistic regression model with predictor $r_\sigma$.
This conditional probability depends only on the observed covariates $X_\sigma$ (through $r_\sigma$) and the parameter of interest $\beta_0$, and does not depend on any of the fixed effects $(\lambda_{im}, \delta_{jm})$ contained in $\mathbf{F}$.
This allows for consistent estimation of $\beta_0$ without needing to estimate the $2NM$ incidental parameters, thereby circumventing the incidental parameter problem.

From Theorem \ref{thm:directed_II_sufficiency}, identification of $\beta_0$ can be established for any cutoff $m$ under standard rank conditions.
Specifically, $\beta_0$ is identified if the expected outer product of the score for the conditional likelihood has full rank:
$$
\beta_0 = \left(E[r_\sigma^\prime r_\sigma]\right)^{-1} 
                E[r_\sigma^\prime \Lambda^{-1}(P\left(\left. \overline Z_\sigma(m) = 1 \right| \overline Z_\sigma(m) \in \{-1,+1\}, X_\sigma\right))].
$$
This demonstrates that $\beta_0$ can be recovered from the observable data distribution $(\mathbf{Y}, \mathbf{X})$ without knowledge of $\mathbf{F}$.

\section{Estimation and inference}\label{sec:estimationII}

We construct and analyze estimators that properly account for the incidental parameters $\mathbf{F}$ building on the identification result in Theorem \ref{thm:directed_II_sufficiency}.
First, we show how to form conditional likelihood contributions using the tetrad approach.
Second, we propose two distinct estimators (ETLE and PTLE) that aggregate information across outcome categories differently.
Finally, we establish consistency properties and discuss computation and inference.

Following Theorem \ref{thm:directed_II_sufficiency}, we construct estimators based on the conditional probabilities that are free of incidental parameters.
We define a tetrad $\sigma$ as informative under the cutoff $m$ when $\overline Z_\sigma(m) \in \{-1,1\}$.
Let $S_{m \sigma} = \mathbf{1}\{\overline Z_\sigma(m) \in \{-1,1\}\}$ indicate this event.
We define:
\begin{equation}
q^*_{mN} = \sum_{\sigma \in \Sigma} S_{m \sigma},
\end{equation}
which counts the total number of informative tetrads in the sample for cutoff $m$.
Here $q_N = |\Sigma|$ denotes the total number of possible tetrads.
The corresponding population proportion is:
\begin{equation}
	p_{m N} = \frac{E[q_{m N}^*]}{q_N}, \label{eq:expectation_proportion}
\end{equation}

For each tetrad-cutoff pair $(\sigma,m)$ that is informative ($S_{m\sigma}=1$), the conditional probability $P(\overline{Z}_\sigma(m)=1 | S_{m\sigma}=1, X_\sigma) = \Lambda(r_\sigma'\beta_0)$.
Let $y^*_{\sigma m} = \mathbf{1}\{\overline{Z}_\sigma(m)=1\}$ denote the binary outcome variable (0 or 1) corresponding to this conditional probability.
The conditional log-likelihood contribution is:
\begin{equation}
	l_{m \sigma}(\beta) = S_{m \sigma}
	\left[
	y^*_{\sigma m} \log \Lambda (r'_\sigma \beta)
	+
	(1 - y^*_{\sigma m}) \log (1 - \Lambda (r'_\sigma \beta))
	\right].
	\label{eq:CLL_contribution_sigma_m}
\end{equation}
This is the standard log-likelihood contribution for a binary logit model applied to the informative tetrad-cutoff pair.
The corresponding score and Hessian contributions are standard:
\begin{align}
	s_{m \sigma}(\beta) &= S_{m \sigma} \left(y^*_{\sigma m} - \Lambda(r'_\sigma \beta)\right) r_\sigma, \label{eq:scoreII} \\
	H_{m \sigma}(\beta) &= - S_{m \sigma} \Lambda (r'_\sigma \beta) (1-\Lambda (r'_\sigma \beta)) r_\sigma r'_\sigma.
\label{eq:HessianII}
\end{align}

For a given cutoff $m$, we can aggregate these contributions over all tetrads:
\begin{equation}
	L_{m N}(\beta) = \sum_{\sigma \in \Sigma} l_{m \sigma} (\beta).
\label{eq:sample_objective_function_m}
\end{equation}
This $L_{m N}(\beta)$ would be the objective function for a binary dyadic model at cutoff $m$, as in \textcite{jochmansSemiparametricAnalysisNetwork2018}.
Similarly, let $s_{mN}(\beta) = \sum_{\sigma \in \Sigma} s_{m\sigma}(\beta)$ and $H_{mN}(\beta) = \sum_{\sigma \in \Sigma} H_{m\sigma}(\beta)$ be the aggregated score and Hessian for cutoff $m$.

\subsection{Two estimation strategies}

We propose two estimators that differ in how they combine information across the different cutoffs $m=1, \dots, M$.
Both estimators maximize an objective function based on the conditional log-likelihood contributions $l_{m\sigma}(\beta)$ over the parameter space $B$.

The first estimator assigns equal weight to the average log-likelihood from each cutoff:
\begin{equation}
	\check\beta_N = \text{argmax}_{\beta \in B} \sum_{m=1}^M \frac{L_{m N}(\beta)}{q^{*}_{mN}}.
\label{eq:sample_objective_function_a}
\end{equation}
This approach effectively treats each cutoff's contribution equally after normalizing by the amount of information available for that cutoff ($q^*_{mN}$).
We refer to this estimator as the equally-weighted tetrad logit estimator (ETLE).
Note that this approach is analogous to analyzing the binary outcome derived from each cutoff $m$ separately using the method of \textcite{jochmansSemiparametricAnalysisNetwork2018}.
It then gives equal weight to the information obtained from each analysis.
We include ETLE as a natural benchmark reflecting this perspective.

The second estimator pools all informative tetrad-cutoff pairs into a single objective function:
\begin{equation}
	\widehat \beta_N = \text{argmax}_{\beta \in B} \sum_{m=1}^M L_{m N}(\beta) = \text{argmax}_{\beta \in B} \sum_{m=1}^M \sum_{\sigma \in \Sigma} l_{m \sigma}(\beta).
\label{eq:sample_objective_function}
\end{equation}
This estimator implicitly weights each cutoff $m$ by the number of informative tetrads $q_{mN}^*$ available for that cutoff.
Cutoffs with more informative tetrads contribute more terms to the sum and thus have a greater influence on the estimate.
We refer to this estimator as the pooled tetrad logit estimator (PTLE).

The fundamental difference lies in how they weight information across cutoffs: ETLE treats each cutoff equally, regardless of the number of informative tetrads, while PTLE allows cutoffs with more informative tetrads to have proportionally greater influence.
ETLE may seem natural in applied work, as it parallels the strategy of separately estimating binary choice models and then combining results.
However, this approach requires each cutoff to provide sufficient identifying information.
The PTLE, in contrast, allows cutoffs with more informative tetrads to have greater influence on the estimation.
By pooling information across all cutoffs before estimation, PTLE can leverage strengths in some categories to compensate for weaknesses in others.
As we will show, this distinction has important implications for consistency under sparse networks or when certain outcome categories are rare.

To analyze the large-sample properties of these estimators, we maintain the following assumptions, adapted from the literature on network formation models.

\begin{assumption}[Sampling]\label{a:sampling_II}
Assume nodes $i \in \mathbb{N}_N$, along with their associated characteristics and fixed effects $F_i$, are independently drawn from a common population distribution.
\end{assumption}
This assumption allows for dependence between dyadic outcomes $Y_{ij}$ that share nodes (through shared fixed effects $F_i, F_j$ or covariates).

\begin{assumption}[Parameter space]\label{a:parameter_II}
The true parameter vector $\beta_0$ is interior to a compact set $B \subset \mathbb{R}^{k}$.
\end{assumption}
This is a standard regularity condition, where $k$ is the dimension of $\beta_0$.

\begin{assumption}[Moments]\label{a:moments_II}
There exists a constant $C<\infty$ such that the second moment of the covariates is uniformly bounded: $E[\|X_{ij}\|^2] < C$ for all dyads $(i,j)$.
\end{assumption}
This ensures that variances and covariances involving covariates are well-behaved.
The consistency of the estimators depends on sufficient information being available from the informative tetrads, which relates to the network's sparsity structure at different cutoffs. We define two scenarios.

\begin{assumption}\label{a:identification_II_strong}
	For every cutoff $m \in \{1, \cdots, M\}$:
	\begin{enumerate}
		\item The expected number of informative tetrads grows sufficiently fast relative to $N$: $N p_{m N} \rightarrow \infty$ as $N \rightarrow \infty$.
		\item The Hessian matrix $\mathcal{H}_m = \lim_{N \rightarrow \infty} (q_N p_{m N})^{-1} E\left[ H_{mN}(\beta_0)\right]$ exists and has full rank $k$.
	\end{enumerate}
\end{assumption}
This assumption implies that each cutoff $m$ individually provides enough information to identify $\beta_0$ asymptotically. 
It corresponds to Assumption 4 in \textcite{jochmansSemiparametricAnalysisNetwork2018} applied to each $m$.

\begin{assumption}\label{a:identification_II_weak}
	Define the pooled probability $p_N = \sum_{m=1}^M p_{mN}$. Then:
	\begin{enumerate}
		\item The total expected number of informative tetrad-cutoff pairs grows sufficiently fast: $N p_N \rightarrow \infty$ as $N \rightarrow \infty$.
		\item The Hessian matrix pooled across all cutoffs, $\mathcal{H} = \lim_{N \rightarrow \infty} (q_N p_{N})^{-1} \sum_{m=1}^M E\left[H_{mN}(\beta_0)\right]$, exists and has full rank $k$.
	\end{enumerate}
\end{assumption}
This assumption only requires sufficient identifying information when pooling across all cutoffs. 
It is weaker than Assumption \ref{a:identification_II_strong}. 
It holds if at least one cutoff $m$ satisfies the conditions of Assumption \ref{a:identification_II_strong}.

The weaker identification condition is particularly relevant for ordered data, where some outcome categories might be rare (e.g., ``best friend'' in a friendship network). 
This sparsity can lead to few informative tetrads ($q^*_{mN}$ small, $p_{mN}$ small) for cutoffs associated with these rare categories, potentially violating Assumption \ref{a:identification_II_strong}. 
Assumption \ref{a:identification_II_weak} allows us to proceed if pooling with more common categories provides enough overall information.

\begin{thm}[Consistency]\label{thm:consistency_II}
	Suppose Assumptions \ref{a:likelihood_II}--\ref{a:moments_II} hold. Then:
	\begin{enumerate}
		\item Under Assumption \ref{a:identification_II_strong}, both the ETLE $\check\beta_N$ and the PTLE $\widehat \beta_N$ converge in probability to $\beta_0$ as $N \rightarrow \infty$.
		\item Under Assumption \ref{a:identification_II_weak}, the PTLE $\widehat \beta_N$ converges in probability to $\beta_0$ as $N \rightarrow \infty$.
	\end{enumerate}
\end{thm}
\begin{proof}See Appendix \ref{sec:proofs}, page \pageref{proof:consistency_II}.\end{proof}

This result highlights the key difference between the estimators: PTLE achieves consistency under weaker identification requirements.

\subsection{Inference}

Constructing valid standard errors requires analyzing the asymptotic distribution of the PTLE estimator.
This analysis involves examining the sum of score contributions, $\sum_{m=1}^M s_{mN}(\beta_0)$, which exhibits dependencies due to tetrads sharing common nodes.
The techniques required are analogous to those developed for binary dyadic models, notably in \textcite{jochmansSemiparametricAnalysisNetwork2018}, typically involving U-statistic theory or projection arguments to handle the network dependence structure.

Define the pooled score and Hessian functions aggregated across all cutoffs and tetrads as:
\begin{align*}
    s_N(\beta) &= \sum_{m=1}^{M} \sum_{\sigma \in \Sigma} s_{m \sigma}(\beta), \\
    H_N(\beta) &= \sum_{m=1}^{M} \sum_{\sigma \in \Sigma} H_{m \sigma}(\beta).
\end{align*}

We need to account for the network dependence structure when constructing the variance estimator.
Let $\Sigma_{ij}$ denote the set of all tetrads $\sigma \in \Sigma$ that include the specific dyad $(i,j)$ among their four constituent dyads.
The dyad-level score sum is:
\[
v_{ij}(\beta) = \sum_{m=1}^{M} \sum_{\sigma \in \Sigma_{ij}} s_{m \sigma}(\beta).
\]

The matrix $\Upsilon_N(\beta)$, formed by summing the outer products of these dyad-level terms, is central to characterizing the asymptotic variance:
\[
\Upsilon_N(\beta) = \sum_{i=1}^N \sum_{j \ne i} v_{ij}(\beta) v_{ij}(\beta)'.
\]

We conjecture that under suitable regularity conditions, including a stronger moment assumption requiring bounded 6th moments of the covariates, the PTLE estimator would have an asymptotic normal distribution.
The asymptotic variance takes the sandwich form $\mathcal{H}^{-1} \Upsilon \mathcal{H}^{-1}$, where $\mathcal{H}$ is the limiting scaled expected pooled Hessian defined in Assumption~\ref{a:identification_II_weak}.(2), and $\Upsilon$ represents the limiting scaled variance component.

A consistent estimator for this asymptotic variance can be constructed using the empirical sandwich formula:
\begin{equation} \label{eq:avar_estimator}
\widehat{\Omega}_N = \left( H_N(\widehat{\beta}_N) \right)^{-1} \Upsilon_N(\widehat{\beta}_N) \left( H_N(\widehat{\beta}_N) \right)^{-1}.
\end{equation}

This variance estimator accounts for the network dependence structure by grouping score contributions by dyads rather than by tetrads, reflecting the fact that the fundamental source of dependence in the network is dyads sharing nodes.
The standard errors for each parameter estimate are then given by the square roots of the diagonal elements of $\widehat{\Omega}_N$.

The construction of this variance estimator parallels the approach developed for binary dyadic models but extends it to accommodate multiple cutoffs in the ordered outcome setting.
While a formal proof would follow similar techniques to those in \textcite{jochmansSemiparametricAnalysisNetwork2018}, our focus here is on providing a practical method for valid inference that accounts for the network dependency structure.
Our simulation study in Section \ref{sec:simulation} demonstrates that this variance estimator performs well in finite samples, even when the network is sparse or when certain outcome categories are rare.

\subsection{Computation}

The PTLE estimator $\widehat{\beta}_N$ can be computed efficiently using standard statistical software. The procedure involves constructing a dataset of informative tetrad-cutoff pairs and running a standard binary logit regression.

\begin{enumerate}
    \item Identify informative tetrads:
    \begin{enumerate}
        \item Generate all $q_N$ possible tetrads $\sigma = (i_1, i_2, j_1, j_2)$ of four distinct nodes.
        \item For each tetrad $\sigma$ and for each cutoff $m \in \{1, \dots, M\}$:
        \begin{itemize}
            \item Construct $D_{i_r,j_s}(m) = \mathbf{1}\{Y_{i_r,j_s} \ge m\}$ for all four dyads.
            \item Calculate the tetrad difference $\overline{Z}_{\sigma}(m) = \frac{1}{2}((D_{i_1,j_1}(m)-D_{i_1,j_2}(m))-(D_{i_2,j_1}(m)-D_{i_2,j_2}(m)))$.
            \item Determine if the pair $(\sigma, m)$ is informative: $S_{m\sigma} = \mathbf{1}\{\overline{Z}_{\sigma}(m) \in \{-1, +1\}\}$.
            \item If $S_{m\sigma}=1$, record the binary outcome $y^*_{\sigma m} = \mathbf{1}\{\overline{Z}_{\sigma}(m) = +1\}$.
        \end{itemize}
    \end{enumerate}

    \item Estimate pooled logit model:
    \begin{enumerate}
        \item Create a dataset containing one observation for each informative pair $(\sigma, m)$ (i.e., where $S_{m\sigma} = 1$). This dataset will have $\sum_{m=1}^M q^*_{mN}$ observations.
        \item For each observation $(\sigma, m)$ in this dataset:
        \begin{itemize}
            \item The dependent variable is the binary outcome $y^*_{\sigma m}$.
            \item The vector of regressors is $r_{\sigma} = (X_{i_1,j_1} - X_{i_1,j_2}) - (X_{i_2,j_1} - X_{i_2,j_2})$.
        \end{itemize}
        \item Estimate a standard binary logit model of $y^*_{\sigma m}$ on $r_{\sigma}$. In \texttt{R}:\\ \texttt{model <- glm(ystar $\sim$ r - 1, family = binomial(link = "logit")}.
        \item The estimated coefficient vector for $r_\sigma$ from this regression is the PTLE $\widehat{\beta}_N$.
    \end{enumerate}

    \item Compute standard errors:
    \begin{enumerate}
        \item Calculate the pooled Hessian matrix $H_N(\widehat{\beta}_N)$ from the logit estimation. 
        \item For each dyad $(i,j)$, identify all tetrads $\sigma \in \Sigma_{ij}$ that include this dyad. 
        \item For each dyad, compute the score sum \(v_{ij}(\widehat{\beta}_N) = \sum_{m=1}^{M} \sum_{\sigma \in \Sigma_{ij}} s_{m \sigma}(\widehat{\beta}_N),\) where $s_{m \sigma}(\widehat{\beta}_N) = S_{m \sigma} (y^*_{m \sigma} - \Lambda(r'_{\sigma} \widehat{\beta}_N)) r_{\sigma}$.
        \item Form the matrix $\Upsilon_N(\widehat{\beta}_N) = \sum_{i=1}^N \sum_{j \ne i} v_{ij}(\widehat{\beta}_N) v_{ij}(\widehat{\beta}_N)'$.
        \item Compute the variance estimator $\widehat{\Omega}_N = (H_N(\widehat{\beta}_N))^{-1} \Upsilon_N(\widehat{\beta}_N) (H_N(\widehat{\beta}_N))^{-1}$.
        \item The standard errors for each coefficient in $\widehat{\beta}_N$ are the square roots of the diagonal elements of $\widehat{\Omega}_N$.
    \end{enumerate}
\end{enumerate}

The ETLE estimator $\check{\beta}_N$ can be computed similarly using a weighted logit approach.
Modify step 2 by applying observation weights $w_{\sigma m} = 1/q_{mN}^*$ to each informative tetrad-cutoff pair $(\sigma, m)$ during the logit estimation. 
In \texttt{R}, we use the \texttt{weights} argument in \texttt{glm}.

\section{Simulation study}
\label{sec:simulation}

We use a series of numerical experiments to evaluate the performance of our proposed estimators.
Our design closely follows the setup described in \textcite{jochmansSemiparametricAnalysisNetwork2018}.
Throughout our simulation study, we consider the ordered outcome case with $M=2$, meaning outcomes can take values in $\{0,1,2\}$.
We have a single regressor defined as
\begin{equation*}
    X_{ij} = (W_i - W_j)^2,
\end{equation*}
where $W_i$ is a Bernoulli random variable with parameter $\frac{1}{2}$.
This dyad-level covariate depends on the characteristics $(W_i,W_j)$ of both nodes.
The node-specific thresholds are functions of the sample size and are constructed as
\begin{equation*}
    \lambda_{im} = \lambda_{m} + \lambda_{i} = \lambda_m + \frac{N-i}{N-1} C_N,
\end{equation*}
with $\lambda_{im} = \delta_{im}$ imposed for symmetry.
The term $\lambda_m$ represents the common threshold component, and $\lambda_i$ captures the node-specific heterogeneity in thresholds, with the parameter $C_N$ controlling this heterogeneity's magnitude.
We simulate networks with two sample sizes $N \in \{25, 50\}$ across different sparsity levels \[C_N \in \{0, \log(\log(N)), \log(N)^{1/2}, \log(N)\},\] with larger values creating increasingly sparse networks.
When $C_N = 0$, there is no node-specific heterogeneity in thresholds, resulting in a dense network.
As $C_N$ increases, the variation in thresholds across nodes increases, leading to sparser networks with fewer links qualifying for higher categories.

\subsection{Homogeneous thresholds}

In our first set of simulations, we focus on homogeneous threshold structures where all nodes share the same base threshold values $\lambda_m$.
With $M=2$, we have two thresholds $\lambda_1$ and $\lambda_2$ that determine the boundaries between the three outcome categories (0, 1, and 2).
Our initial setting uses the threshold values $\lambda_1 = 0$ and $\lambda_2 = 1$.
For each generated network, we implement our binarization approach using two cutoffs: $m_{\text{cutoff}} = 1$ and $m_{\text{cutoff}} = 2$.
These cutoffs create binary indicators based on whether the outcome meets or exceeds the specified threshold: $D_{ij}(m) = \mathbf{1}\{Y_{ij} \geq m\}$.

Table \ref{tab:degree_distribution} presents summary statistics of the degree distributions for the binarized networks across different values of $C_N$ and for both sample sizes $N$.
The degree distributions, representing the average number of connections per node for the binarized networks $D_{ij}(m)$, show that networks become sparser as $C_N$ increases.
This is reflected in the declining mean degrees for both cutoffs.

Table \ref{tab:simulation_results} presents the simulation results comparing our two proposed estimators, ETLE and PTLE, based on 1000 Monte Carlo replications.
We report the mean, median, standard deviation (std), and interquartile range (iqr) of the estimates of $\beta_0=1$ for each configuration.
Both estimators perform similarly in dense networks ($C_N = 0$), but their performance diverges as sparsity increases.
The performance gap between ETLE and PTLE widens in sparser networks, with PTLE consistently showing less bias.
The superior performance of PTLE is less pronounced for the larger sample size of $N=50$.
This is expected, as the larger network size generates a sufficient number of informative tetrads across thresholds to improve estimation for both models, even in sparser settings.
Both estimators exhibit increased bias when $C_N = \log(N)$, consistent with our observation that this setting represents an extremely sparse network design.

We also compare our estimators against standard binary choice models using tetrad-differencing with different single cutoffs.
The binary model with $m_{\text{cutoff}} = 1$ yields results similar to PTLE, benefiting from more informative tetrads at this threshold.
In contrast, the binary model with $m_{\text{cutoff}} = 2$ produces more biased estimates due to fewer informative tetrads.
This comparison shows that the performance of binary models depends critically on the chosen cutoff.
Moreover, selecting only one threshold, even the most informative one, is inefficient as it ignores information from the rest of the outcome distribution. This inefficiency will be much more pronounced under combinations of smaller $N$, more sparsity and a larger number of categories $M$.
The proposed ordered approach, on the other hand, is more principled, as it systematically leverages the full structure of the data.

To further examine the impact of threshold magnitudes, we conduct a second set of simulations with $\lambda_1 = 0$ and $\lambda_2 = 2$.
This wider threshold gap reduces the number of informative tetrads in the higher categories.
Table \ref{tab:degree_distribution_2} shows that for $m_{\text{cutoff}} = 2$, the mean degree is approximately half that observed with the narrower threshold range.
Specifically, comparing Tables \ref{tab:degree_distribution} and \ref{tab:degree_distribution_2} for $C_N = 0$, the mean degree drops from 0.385 to 0.194.

The corresponding estimation results in Table \ref{tab:simulation_results_2} reveal that especially in terms of standard deviations and interquartile range ETLE exhibits a more pronounced deterioration compared to PTLE. This is a direct consequence of ETLE's equal-weighting scheme.
It gives the sparser, less-informative higher category the same influence on the final estimate as the denser lower category, which degrades its performance. As a further results of this, the mean bias of ETLE deteriorates quite markedly under $N=25$ and $C_N=\log(N)$ when the network is small and sparse---with almost no change in PTLE.
The results also confirm that binary choice models become increasingly sensitive to cutoff selection as the threshold range widens.
Inappropriate cutoffs substantially increase both bias and variability, particularly in sparse networks.

To visualize the effects of network sparsity and threshold magnitudes on the mean degrees for different cutoff points, we conduct additional simulations for $N=25$ (1000 replications) varying the maximum threshold $\lambda_2$ while keeping $\lambda_1 = 0$ fixed.
Figure \ref{fig:mean_degree_all} illustrates the effect on the mean network degree of increasing $\lambda_2$ for two sparsity levels.
Panel (a) shows results for $C_N=0$ while panel (b) shows results for $C_N=\log(N)^{1/2}$. We note that for a given $C_N$ the mean degrees for the first cutoff remain relatively stable since $\lambda_1$ is fixed. Figure \ref{fig:mean_degree_all} conveys two messages: first, inducing sparsity leads to a uniform decrease on the mean degrees across both cutoffs. Second, within a given sparsity setting, increasing $\lambda_2$ further decreases the mean degree for the second cutoff. Indeed, for $C_N=0$ it is clear that the mean degree for $m_{\text{cutoff}}=2$ approaches zero---and becomes almost zero for $\lambda_2 > 3$ .

Figure \ref{fig:cmle_all} presents the mean estimates across replications for the ETLE and PTLE estimators, across different values of $\lambda_2$.
Consistent with the observation that mean degrees for $m_{\text{cutoff}}=2$ become very low for $\lambda_2 > 3$, we restrict our reporting to $\lambda_2 \leq 3$, as the low number of informative tetrads for large $\lambda_2$ leads to numerical instability and unreliable estimates.
We observe that in both cases the difference between ETLE and PTLE widens as the mean degree for the second cutoff gets close to zero, implying a progressively lower number of informative tetrads for this cutoff. For $C_N=0$ this effect kicks in as $\lambda_2$ becomes larger. For $C_N=\log(N^{1/2})$, on the other hand, the worsening of the ETLE estimator begins at even lower threshold values, becoming progressively worse. This is not surprising in light of panel (b) of Figure \ref{fig:cmle_diff_all}: in this scenario, the mean degree for the second cutoff is already almost equal to zero. Importantly no such worsening of estimator performance is observed for PTLE. As $\lambda_2$ increases, the mean estimates from PTLE remain consistently close to the true parameter value $\beta_0 = 1$.  ETLE, on the other hand, clearly requires a minimum mean degree across all threshold to be reliable.
This pattern provides strong visual confirmation of our main theoretical result.
PTLE's robustness stems from its pooling mechanism, which naturally gives more influence to categories with more information.
In contrast, ETLE's equal-weighting scheme makes it vulnerable to instability when any single category becomes sparse.

\begin{table}[ht!]
    \centering
    \begin{tabularx}{\textwidth}{l *{6}{>{\centering\arraybackslash}X}}
        \toprule
        \( C_N \) & Mean & $Q_{0.25}$ & $Q_{0.5}$ & $Q_{0.75}$ & Min. & Max. \\
        \midrule
        \multicolumn{7}{c}{\( N = 25,\; m_{\text{cutoff}} = 1 \)} \\
        0                       & 0.616 & 0.555 & 0.618 & 0.681 & 0.417 & 0.804 \\
        \(\log(\log(N))\)      & 0.354 & 0.273 & 0.350 & 0.430 & 0.140 & 0.596 \\
        \(\log(N^{1/2})\)      & 0.249 & 0.161 & 0.237 & 0.325 & 0.051 & 0.505 \\
        \(\log(N)\)            & 0.114 & 0.040 & 0.092 & 0.172 & 0.000 & 0.347 \\
        \multicolumn{7}{c}{\( N = 25,\; m_{\text{cutoff}} = 2 \)} \\
        0                       & 0.385 & 0.321 & 0.382 & 0.448 & 0.198 & 0.586 \\
        \(\log(\log(N))\)      & 0.180 & 0.118 & 0.172 & 0.235 & 0.032 & 0.377 \\
        \(\log(N^{1/2})\)      & 0.119 & 0.060 & 0.107 & 0.166 & 0.004 & 0.306 \\
        \(\log(N)\)            & 0.052 & 0.002 & 0.038 & 0.082 & 0.000 & 0.203 \\
        \midrule
        \multicolumn{7}{c}{\( N = 50,\; m_{\text{cutoff}} = 1 \)} \\
        0                       & 0.615 & 0.570 & 0.616 & 0.661 & 0.458 & 0.767 \\
        \(\log(\log(N))\)      & 0.317 & 0.244 & 0.312 & 0.386 & 0.122 & 0.545 \\
        \(\log(N^{1/2})\)      & 0.222 & 0.141 & 0.210 & 0.294 & 0.046 & 0.464 \\
        \(\log(N)\)            & 0.082 & 0.021 & 0.058 & 0.127 & 0.000 & 0.292 \\
        \multicolumn{7}{c}{\( N = 50,\; m_{\text{cutoff}} = 2 \)} \\
        0                       & 0.385 & 0.339 & 0.384 & 0.429 & 0.234 & 0.542 \\
        \(\log(\log(N))\)      & 0.157 & 0.106 & 0.150 & 0.202 & 0.031 & 0.332 \\
        \(\log(N^{1/2})\)      & 0.105 & 0.056 & 0.094 & 0.145 & 0.005 & 0.274 \\
        \(\log(N)\)            & 0.037 & 0.001 & 0.022 & 0.056 & 0.000 & 0.166 \\
        \bottomrule
    \end{tabularx}
    \caption{Degree distributions for simulated data: $\lambda_m = \{0,1\}$.}
    \label{tab:degree_distribution}
\end{table}
\FloatBarrier

\begin{table}[ht!]
    \centering
    \begin{tabularx}{\textwidth}{l *{8}{>{\centering\arraybackslash}X}}
        \toprule
        & \multicolumn{4}{c}{\( N = 25 \)} & \multicolumn{4}{c}{\( N = 50 \)} \\
        \cmidrule(lr){2-5} \cmidrule(lr){6-9}
        & ETLE & PTLE & Binary (1) & Binary (2) & ETLE & PTLE & Binary (1) & Binary (2) \\
        \midrule
        \multicolumn{9}{c}{\( C_N = 0 \)} \\
        mean    & 1.009 & 1.009 & 1.008 & 1.015 & 1.001 & 1.001 & 1.004 & 0.999 \\
        median  & 1.006 & 1.007 & 1.005 & 1.009 & 1.000 & 1.000 & 1.004 & 0.998 \\
        std     & 0.166 & 0.166 & 0.188 & 0.185 & 0.080 & 0.080 & 0.089 & 0.087 \\
        iqr     & 0.223 & 0.224 & 0.251 & 0.238 & 0.102 & 0.103 & 0.123 & 0.115 \\
        \midrule
        \multicolumn{9}{c}{\( C_N = \log(\log(N)) \)} \\
        mean    & 1.019 & 1.015 & 1.015 & 1.042 & 1.005 & 1.004 & 1.003 & 1.009 \\
        median  & 1.008 & 1.008 & 1.010 & 1.021 & 1.007 & 1.006 & 1.003 & 1.008 \\
        std     & 0.306 & 0.292 & 0.306 & 0.434 & 0.097 & 0.092 & 0.094 & 0.126 \\
        iqr     & 0.271 & 0.255 & 0.275 & 0.349 & 0.133 & 0.124 & 0.121 & 0.167 \\
        \midrule
        \multicolumn{9}{c}{\( C_N = \log(N^{1/2}) \)} \\
        mean    & 1.021 & 1.017 & 1.018 & 1.094 & 1.002 & 1.002 & 1.002 & 1.006 \\
        median  & 1.008 & 1.007 & 1.005 & 1.013 & 0.999 & 0.997 & 0.998 & 1.002 \\
        std     & 0.259 & 0.243 & 0.247 & 0.760 & 0.119 & 0.112 & 0.114 & 0.155 \\
        iqr     & 0.344 & 0.313 & 0.325 & 0.441 & 0.169 & 0.158 & 0.152 & 0.219 \\
        \midrule
        \multicolumn{9}{c}{\( C_N = \log(N) \)} \\
        mean    & 1.199 & 1.180 & 1.184 & 1.990 & 1.015 & 1.011 & 1.011 & 1.034 \\
        median  & 1.036 & 1.035 & 1.032 & 1.073 & 1.012 & 1.003 & 1.001 & 1.018 \\
        std     & 1.135 & 1.129 & 1.153 & 2.744 & 0.207 & 0.186 & 0.185 & 0.289 \\
        iqr     & 0.534 & 0.477 & 0.461 & 0.771 & 0.278 & 0.253 & 0.241 & 0.370 \\
        \bottomrule
    \end{tabularx}
    \caption{Simulation results: $\lambda_m = \{0,1\}$. Binary (1) and Binary (2) refer to standard binary choice models estimated using the tetrad-differencing technique, with cutoff thresholds $m_{\text{cutoff}}=1$ and $m_{\text{cutoff}}=2$, respectively.}
    \label{tab:simulation_results}
\end{table}
\FloatBarrier

\begin{table}[ht!]
    \centering
    \begin{tabularx}{\textwidth}{l *{6}{>{\centering\arraybackslash}X}}
        \toprule
        \( C_N \) & Mean & $Q_{0.25}$ & $Q_{0.5}$ & $Q_{0.75}$ & Min. & Max. \\
        \midrule
        \multicolumn{7}{c}{\( N = 25,\; m_{\text{cutoff}} = 1 \)} \\
        0                       & 0.616 & 0.555 & 0.618 & 0.681 & 0.417 & 0.804 \\
        \(\log(\log(N))\)      & 0.354 & 0.273 & 0.350 & 0.430 & 0.140 & 0.596 \\
        \(\log(N^{1/2})\)      & 0.249 & 0.161 & 0.237 & 0.325 & 0.051 & 0.505 \\
        \(\log(N)\)            & 0.114 & 0.040 & 0.092 & 0.172 & 0.000 & 0.347 \\
        \multicolumn{7}{c}{\( N = 25,\; m_{\text{cutoff}} = 2 \)} \\
        0                       & 0.194 & 0.141 & 0.190 & 0.243 & 0.053 & 0.363 \\
        \(\log(\log(N))\)      & 0.078 & 0.038 & 0.071 & 0.110 & 0.001 & 0.216 \\
        \(\log(N^{1/2})\)      & 0.050 & 0.008 & 0.041 & 0.075 & 0.000 & 0.171 \\
        \(\log(N)\)            & 0.022 & 0.000 & 0.002 & 0.038 & 0.000 & 0.113 \\
        \midrule
        \multicolumn{7}{c}{\( N = 50,\; m_{\text{cutoff}} = 1 \)} \\
        0                       & 0.616 & 0.570 & 0.617 & 0.662 & 0.459 & 0.769 \\
        \(\log(\log(N))\)      & 0.317 & 0.243 & 0.311 & 0.386 & 0.122 & 0.544 \\
        \(\log(N^{1/2})\)      & 0.223 & 0.143 & 0.211 & 0.295 & 0.046 & 0.465 \\
        \(\log(N)\)            & 0.082 & 0.021 & 0.058 & 0.127 & 0.000 & 0.292 \\
        \multicolumn{7}{c}{\( N = 50,\; m_{\text{cutoff}} = 2 \)} \\
        0                       & 0.194 & 0.156 & 0.192 & 0.230 & 0.079 & 0.329 \\
        \(\log(\log(N))\)      & 0.067 & 0.038 & 0.062 & 0.092 & 0.001 & 0.180 \\
        \(\log(N^{1/2})\)      & 0.044 & 0.019 & 0.038 & 0.064 & 0.000 & 0.149 \\
        \(\log(N)\)            & 0.015 & 0.000 & 0.002 & 0.022 & 0.000 & 0.090 \\
        \bottomrule
    \end{tabularx}
    \caption{Degree distributions for simulated data: $\lambda_m = \{0,2\}$.}
    \label{tab:degree_distribution_2}
\end{table}
\FloatBarrier

\begin{table}[ht!]
    \centering
    \begin{tabularx}{\textwidth}{l *{8}{>{\centering\arraybackslash}X}}
        \toprule
        & \multicolumn{4}{c}{\( N = 25 \)} & \multicolumn{4}{c}{\( N = 50 \)} \\
        \cmidrule(lr){2-5} \cmidrule(lr){6-9}
        & ETLE & PTLE & Binary (1) & Binary (2) & ETLE & PTLE & Binary (1) & Binary (2) \\
        \midrule
        \multicolumn{9}{c}{\( C_N = 0 \)} \\
        mean    & 1.008 & 1.006 & 1.008 & 1.024 & 1.001 & 1.002 & 1.003 & 1.002 \\
        median  & 1.007 & 1.002 & 1.005 & 1.020 & 1.001 & 1.003 & 1.008 & 1.000 \\
        std     & 0.178 & 0.172 & 0.188 & 0.255 & 0.082 & 0.080 & 0.089 & 0.109 \\
        iqr     & 0.235 & 0.231 & 0.251 & 0.324 & 0.113 & 0.110 & 0.115 & 0.144 \\
        \midrule
        \multicolumn{9}{c}{\( C_N = \log(\log(N)) \)} \\
        mean    & 1.038 & 1.017 & 1.015 & 1.324 & 1.003 & 1.000 & 1.000 & 1.014 \\
        median  & 1.028 & 1.012 & 1.010 & 1.063 & 1.004 & 1.000 & 1.001 & 1.010 \\
        std     & 0.333 & 0.304 & 0.306 & 1.494 & 0.118 & 0.095 & 0.096 & 0.192 \\
        iqr     & 0.318 & 0.261 & 0.275 & 0.549 & 0.157 & 0.129 & 0.135 & 0.260 \\
        \midrule
        \multicolumn{9}{c}{\( C_N = \log(N^{1/2}) \)} \\
        mean    & 1.027 & 1.015 & 1.018 & 1.804 & 1.002 & 0.998 & 0.998 & 1.021 \\
        median  & 1.016 & 1.005 & 1.005 & 1.019 & 1.004 & 0.996 & 0.993 & 1.007 \\
        std     & 0.310 & 0.244 & 0.247 & 2.530 & 0.150 & 0.111 & 0.112 & 0.257 \\
        iqr     & 0.421 & 0.326 & 0.325 & 0.670 & 0.199 & 0.152 & 0.155 & 0.313 \\
        \midrule
        \multicolumn{9}{c}{\( C_N = \log(N) \)} \\
        mean    & 1.237 & 1.176 & 1.184 & 4.423 & 1.017 & 1.009 & 1.010 & 1.341 \\
        median  & 1.119 & 1.037 & 1.032 & 1.386 & 1.018 & 0.997 & 0.995 & 1.028 \\
        std     & 1.148 & 1.125 & 1.153 & 4.682 & 0.258 & 0.185 & 0.187 & 1.627 \\
        iqr     & 0.696 & 0.458 & 0.461 & 8.657 & 0.358 & 0.244 & 0.241 & 0.618 \\
        \bottomrule
    \end{tabularx}
    \caption{Simulation results: $\lambda_m = \{0,2\}$. Binary (1) and Binary (2) refer to standard binary choice models estimated using the tetrad-differencing technique, with cutoff thresholds $m_{\text{cutoff}}=1$ and $m_{\text{cutoff}}=2$, respectively.}
    \label{tab:simulation_results_2}
\end{table}
\FloatBarrier

\begin{figure}[htbp]
\centering
  \begin{subfigure}[t]{0.48\textwidth}
    \includegraphics[width=\textwidth]{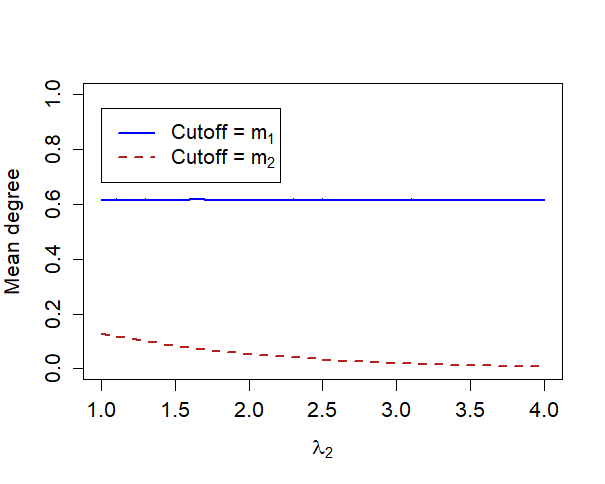}
    \caption{\(C_N = 0\)}
    \label{fig:mean_degree_0}
  \end{subfigure}
  \hfill
  \begin{subfigure}[t]{0.48\textwidth}
    \includegraphics[width=\textwidth]{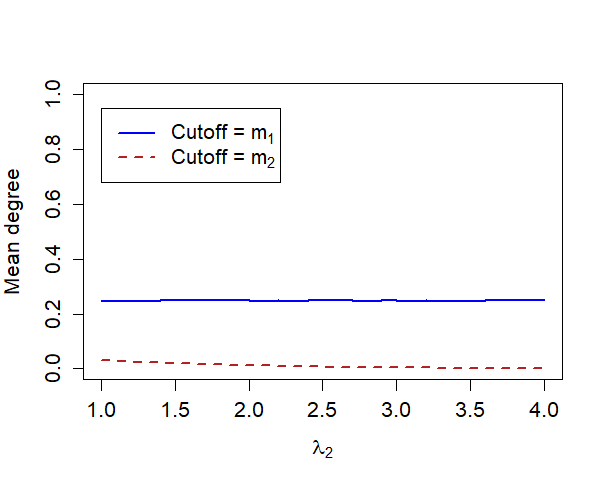}
    \caption{\(C_N = \log(N^{1/2})\)}
    \label{fig:mean_degree_logsqrtn}
  \end{subfigure}
  \caption{Mean degrees across \(\lambda_2\) by cutoff category for \(C_N = 0\) and \(C_N = \log(N^{1/2})\).}
  \label{fig:mean_degree_all}
\end{figure}
\FloatBarrier

\begin{figure}[htbp]
  \centering
  \begin{subfigure}[t]{0.48\textwidth}
    \includegraphics[width=\textwidth]{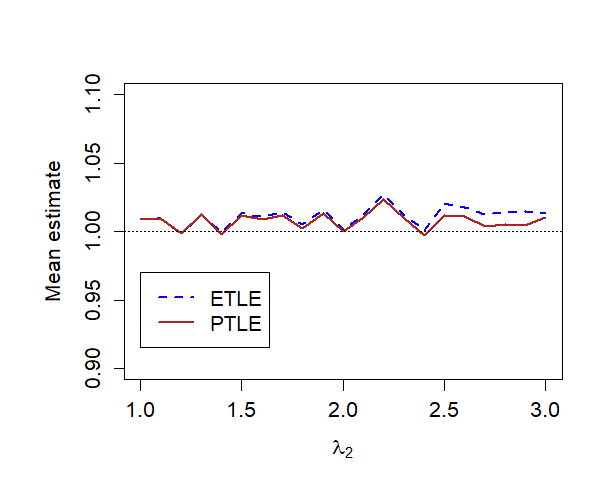}
    \caption{\(C_N = 0\)}
    \label{fig:cmle_0}
  \end{subfigure}
  \hfill
  \begin{subfigure}[t]{0.48\textwidth}
    \includegraphics[width=\textwidth]{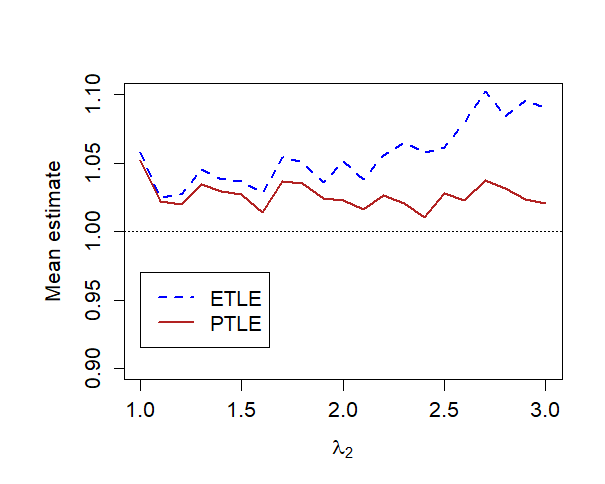}
    \caption{\(C_N = \log(N^{1/2})\)}
    \label{fig:cmle_logsqrt}
  \end{subfigure}
  \caption{Comparison of mean estimates across values of the maximum threshold, \(\lambda_2\), under \(C_N = 0\) and \(C_N = \log(N^{1/2})\) for two estimators.}
  \label{fig:cmle_all}
\end{figure}
\FloatBarrier

\subsection{Heterogeneous thresholds}

To capture heterogeneity in threshold structures across different types of nodes, we introduce variation in the threshold sets $\lambda_m$ based on the binary covariate $W_i$.
This setting reflects environments where individuals differ in their willingness or propensity to form connections—such as in a school context, where some students may form friendships easily while others require stronger commonalities.
We model this by creating two distinct node types, assigning narrower threshold ranges to individuals with $W_i = 0$ and wider ranges to those with $W_i = 1$.
Formally, we define the thresholds as:
\begin{equation}
    \lambda_{m} = \begin{cases} 
       \lambda_{m,0} & \text{if } W_i = 0 \\
       \lambda_{m,1} & \text{if } W_i = 1 
    \end{cases}
\end{equation}
where $\lambda_{m,0}$ and $\lambda_{m,1}$ represent the threshold values for nodes with $W_i = 0$ and $W_i = 1$, respectively.
We consider two configurations: (i) $\lambda_1 = 0$, $\lambda_{2,0} = 0.5$, $\lambda_{2,1} = 1$; and (ii) $\lambda_1 = 0$, $\lambda_{2,0} = 0.5$, $\lambda_{2,1} = 1.5$.
The corresponding degree distributions are presented in Table \ref{tab:degree_distribution_0.5_1} and Table \ref{tab:degree_distribution_0.5_1.5}, which show patterns consistent with the characteristics described in the previous subsection.

To implement these heterogeneous thresholds, we specify the components of the total threshold for a dyad $(i,j)$.
The sender's type-dependent component is $\lambda_{m,0}$ if $W_i = 0$ or $\lambda_{m,1}$ if $W_i = 1$.
The receiver's type-dependent component is determined similarly by $W_j$.
The total threshold for the dyad is the sum of these two type-dependent components and the individual-specific shifts, $\lambda_i + \lambda_j$.
For example, under configuration (i) with $M = 2$, if $W_i = 1$ and $W_j = 0$, the first threshold is $0 + 0 + \lambda_i + \lambda_j$ (since $\lambda_1 = 0$ for both types), while the second threshold is $1 + 0.5 + \lambda_i + \lambda_j$ (using $\lambda_{2,1} = 1$ for node $i$ and $\lambda_{2,0} = 0.5$ for node $j$).

Tables \ref{tab:simulation_results_0.5_1} and \ref{tab:simulation_results_0.5_1.5} present the simulation results under these heterogeneous threshold designs.
Across all scenarios, the PTLE estimator stands out as the best option, regardless of sample size or sparsity level. However, the difference between PTLE and ETLE becomes quite pronounced under sparsity ($C_N=\log(N)$), which is the main setting of interest. The results also confirm that, as expected, estimation performance deteriorates when the maximum threshold for $W_i = 1$ increases from 1.0 to 1.5, implying a sparser network structure for the second cutoff.
The same is also observed for the binary models, the performance of which is highly sensitive to the choice of $m_{\text{cutoff}}$, with poorly chosen thresholds yielding substantial estimation bias, especially in sparser networks.
These results demonstrate that as heterogeneity in connection propensities increases, the choice of a robust estimator becomes even more critical.
PTLE's performance in these settings highlights its suitability for empirical applications where such structural variation is likely to be present.

To systematically examine the role of threshold heterogeneity, we consider a simulation exercise where the maximum threshold $\lambda_{2,1}$ ranges from 1 to 3, while  keeping $\lambda_{2,0} = 0.5$ fixed.
As in the homogeneous threshold analysis, we restrict attention to cases where $\lambda_{2,1} \leq 3$ to ensure reliable inference.
Figure \ref{fig:cmle_diff_all} reports the mean estimates across replications for both ETLE and PTLE as $\lambda_{2,1}$ increases.
In both dense and moderately sparse networks, PTLE maintains estimates close to the true parameter value $\beta_0 = 1$, showing minimal deviation as $\lambda_{2,1}$ expands.
In contrast, ETLE estimates systematically drift away from the true value as threshold heterogeneity increases, with the divergence more pronounced under sparser conditions.
These results reinforce our main theoretical prediction.
Because its pooling approach allows more informative categories to have a greater influence on the final estimate, PTLE is better suited for the heterogeneous and sparse connection patterns common in applied work.

\begin{table}[ht!]
\centering
    \begin{tabularx}{\textwidth}{l *{6}{>{\centering\arraybackslash}X}}
        \toprule
        \( C_N \) & Mean & $Q_{0.25}$ & $Q_{0.5}$ & $Q_{0.75}$ & Min. & Max. \\
        \midrule
        \multicolumn{7}{c}{\( N = 25,\; m_{\text{cutoff}} = 1 \)} \\
        0                      & 0.616 & 0.555 & 0.618 & 0.681 & 0.417 & 0.804 \\
        \(\log(\log(N))\)      & 0.355 & 0.274 & 0.351 & 0.430 & 0.138 & 0.596 \\
        \(\log(N^{1/2})\)      & 0.248 & 0.162 & 0.237 & 0.324 & 0.052 & 0.507 \\
        \(\log(N)\)            & 0.115 & 0.040 & 0.093 & 0.173 & 0.000 & 0.350 \\
        \multicolumn{7}{c}{\( N = 25,\; m_{\text{cutoff}} = 2 \)} \\
        0                      & 0.391 & 0.319 & 0.387 & 0.458 & 0.194 & 0.611 \\
        \(\log(\log(N))\)      & 0.185 & 0.119 & 0.176 & 0.240 & 0.033 & 0.390 \\
        \(\log(N^{1/2})\)      & 0.122 & 0.061 & 0.110 & 0.170 & 0.004 & 0.318 \\
        \(\log(N)\)            & 0.053 & 0.003 & 0.039 & 0.081 & 0.000 & 0.209 \\
        \midrule
        \multicolumn{7}{c}{\( N = 50,\; m_{\text{cutoff}} = 1 \)} \\
        0                      & 0.616 & 0.570 & 0.617 & 0.662 & 0.459 & 0.769 \\
        \(\log(\log(N))\)      & 0.317 & 0.243 & 0.311 & 0.386 & 0.122 & 0.544 \\
        \(\log(N^{1/2})\)      & 0.223 & 0.143 & 0.211 & 0.295 & 0.046 & 0.465 \\
        \(\log(N)\)            & 0.082 & 0.021 & 0.058 & 0.127 & 0.000 & 0.292 \\
        \multicolumn{7}{c}{\( N = 50,\; m_{\text{cutoff}} = 2 \)} \\
        0                      & 0.390 & 0.332 & 0.387 & 0.446 & 0.221 & 0.580 \\
        \(\log(\log(N))\)      & 0.161 & 0.107 & 0.152 & 0.207 & 0.031 & 0.352 \\
        \(\log(N^{1/2})\)      & 0.108 & 0.057 & 0.097 & 0.148 & 0.005 & 0.290 \\
        \(\log(N)\)            & 0.038 & 0.001 & 0.022 & 0.058 & 0.000 & 0.174 \\
        \bottomrule
    \end{tabularx}
    \caption{Degree distributions for simulated data: $\lambda_1 = 0$, $\lambda_{2,0} = 0.5$ and $\lambda_{2,1} = 1$.}
    \label{tab:degree_distribution_0.5_1}
\end{table}
\FloatBarrier

\begin{table}[ht!]
\centering
    \begin{tabularx}{\textwidth}{l *{8}{>{\centering\arraybackslash}X}}
        \toprule
        & \multicolumn{4}{c}{\( N = 25 \)} & \multicolumn{4}{c}{\( N = 50 \)} \\
        \cmidrule(lr){2-5} \cmidrule(lr){6-9}
        & ETLE & PTLE & Binary (1) & Binary (2) & ETLE & PTLE & Binary (1) & Binary (2) \\
        \midrule
        \multicolumn{9}{c}{\( C_N = 0 \)} \\
        mean    & 1.008 & 1.008 & 1.008 & 1.015 & 1.003 & 1.003 & 1.004 & 1.003 \\
        median  & 1.002 & 1.001 & 1.005 & 1.006 & 1.000 & 1.001 & 1.004 & 1.001 \\
        std     & 0.168 & 0.168 & 0.188 & 0.194 & 0.080 & 0.080 & 0.089 & 0.090 \\
        iqr     & 0.221 & 0.220 & 0.251 & 0.249 & 0.108 & 0.107 & 0.123 & 0.121 \\
        \midrule
        \multicolumn{9}{c}{\( C_N = \log(\log(N)) \)} \\
        mean    & 1.016 & 1.014 & 1.015 & 1.052 & 1.004 & 1.003 & 1.003 & 1.008 \\
        median  & 1.010 & 1.008 & 1.010 & 1.022 & 1.006 & 1.006 & 1.003 & 1.008 \\
        std     & 0.305 & 0.303 & 0.306 & 0.574 & 0.098 & 0.092 & 0.094 & 0.131 \\
        iqr     & 0.268 & 0.260 & 0.275 & 0.329 & 0.133 & 0.123 & 0.121 & 0.175 \\
        \midrule
        \multicolumn{9}{c}{\( C_N = \log(N^{1/2}) \)} \\
        mean    & 1.020 & 1.016 & 1.018 & 1.180 & 1.001 & 1.001 & 1.002 & 1.005 \\
        median  & 0.996 & 1.001 & 1.005 & 1.000 & 1.997 & 0.998 & 0.998 & 0.995 \\
        std     & 0.265 & 0.245 & 0.247 & 1.154 & 0.120 & 0.112 & 0.114 & 0.165 \\
        iqr     & 0.335 & 0.310 & 0.325 & 0.456 & 0.166 & 0.155 & 0.152 & 0.219 \\
        \midrule
        \multicolumn{9}{c}{\( C_N = \log(N) \)} \\
        mean    & 1.196 & 1.175 & 1.184 & 2.280 & 1.014 & 1.011 & 1.011 & 1.073 \\
        median  & 1.052 & 1.030 & 1.032 & 1.084 & 1.004 & 1.001 & 1.001 & 1.001 \\
        std     & 1.127 & 1.107 & 1.153 & 3.055 & 0.212 & 0.188 & 0.185 & 0.662 \\
        iqr     & 0.542 & 0.473 & 0.461 & 0.874 & 0.273 & 0.245 & 0.241 & 0.383 \\
        \bottomrule
    \end{tabularx}
    \caption{Simulation results: $\lambda_1 = 0$, $\lambda_{2,0} = 0.5$ and $\lambda_{2,1} = 1$. Binary (1) and Binary (2) refer to standard binary choice models estimated using the tetrad-differencing technique, with cutoff thresholds $m_{\text{cutoff}}=1$ and $m_{\text{cutoff}}=2$, respectively.}
    \label{tab:simulation_results_0.5_1}
\end{table}
\FloatBarrier

\begin{table}[ht!]
    \centering
    \begin{tabularx}{\textwidth}{l *{6}{>{\centering\arraybackslash}X}}
        \toprule
        \( C_N \) & Mean & $Q_{0.25}$ & $Q_{0.5}$ & $Q_{0.75}$ & Min. & Max. \\
        \midrule
        \multicolumn{7}{c}{\( N = 25,\; m_{\text{cutoff}} = 1 \)} \\
        0                      & 0.616 & 0.555 & 0.618 & 0.681 & 0.417 & 0.804 \\
        \(\log(\log(N))\)      & 0.355 & 0.274 & 0.351 & 0.430 & 0.140 & 0.596 \\
        \(\log(N^{1/2})\)      & 0.249 & 0.161 & 0.237 & 0.325 & 0.051 & 0.505 \\
        \(\log(N)\)            & 0.114 & 0.040 & 0.093 & 0.172 & 0.000 & 0.347 \\
        \multicolumn{7}{c}{\( N = 25,\; m_{\text{cutoff}} = 2 \)} \\
        0                      & 0.303 & 0.223 & 0.293 & 0.377 & 0.109 & 0.542 \\
        \(\log(\log(N))\)      & 0.134 & 0.075 & 0.123 & 0.180 & 0.010 & 0.331 \\
        \(\log(N^{1/2})\)      & 0.087 & 0.037 & 0.074 & 0.125 & 0.000 & 0.264 \\
        \(\log(N)\)            & 0.038 & 0.000 & 0.022 & 0.058 & 0.000 & 0.172 \\
        \midrule
        \multicolumn{7}{c}{\( N = 50,\; m_{\text{cutoff}} = 1 \)} \\
        0                      & 0.616 & 0.570 & 0.617 & 0.662 & 0.459 & 0.769 \\
        \(\log(\log(N))\)      & 0.317 & 0.243 & 0.311 & 0.386 & 0.122 & 0.544 \\
        \(\log(N^{1/2})\)      & 0.223 & 0.143 & 0.211 & 0.295 & 0.046 & 0.465 \\
        \(\log(N)\)            & 0.082 & 0.021 & 0.058 & 0.127 & 0.000 & 0.292 \\
        \multicolumn{7}{c}{\( N = 50,\; m_{\text{cutoff}} = 2 \)} \\
        0                      & 0.303 & 0.228 & 0.293 & 0.375 & 0.126 & 0.516 \\
        \(\log(\log(N))\)      & 0.116 & 0.068 & 0.106 & 0.154 & 0.010 & 0.297 \\
        \(\log(N^{1/2})\)      & 0.077 & 0.035 & 0.065 & 0.107 & 0.000 & 0.242 \\
        \(\log(N)\)            & 0.027 & 0.000 & 0.017 & 0.040 & 0.000 & 0.142 \\
        \bottomrule
    \end{tabularx}
    \caption{Degree distributions for simulated data: $\lambda_1 = 0$, $\lambda_{2,0} = 0.5$ and $\lambda_{2,1} = 1.5$.}
    \label{tab:degree_distribution_0.5_1.5}
\end{table}
\FloatBarrier

\begin{table}[ht!]
    \centering
    \begin{tabularx}{\textwidth}{l *{8}{>{\centering\arraybackslash}X}}
        \toprule
        & \multicolumn{4}{c}{\( N = 25 \)} & \multicolumn{4}{c}{\( N = 50 \)} \\
        \cmidrule(lr){2-5} \cmidrule(lr){6-9}
        & ETLE & PTLE & Binary (1) & Binary (2) & ETLE & PTLE & Binary (1) & Binary (2) \\
        \midrule
        \multicolumn{9}{c}{\( C_N = 0 \)} \\
        mean    & 1.010 & 1.008 & 1.008 & 1.037 & 1.004 & 1.004 & 1.004 & 1.007 \\
        median  & 1.008 & 1.004 & 1.005 & 1.019 & 1.002 & 1.004 & 1.004 & 1.004 \\
        std     & 0.178 & 0.175 & 0.188 & 0.362 & 0.084 & 0.083 & 0.089 & 0.111 \\
        iqr     & 0.228 & 0.230 & 0.251 & 0.313 & 0.114 & 0.113 & 0.123 & 0.153 \\
        \midrule
        \multicolumn{9}{c}{\( C_N = \log(\log(N)) \)} \\
        mean    & 1.023 & 1.015 & 1.015 & 1.376 & 1.007 & 1.004 & 1.003 & 1.023 \\
        median  & 1.009 & 1.007 & 1.010 & 1.037 & 1.000 & 1.004 & 1.003 & 1.006 \\
        std     & 0.334 & 0.306 & 0.306 & 1.677 & 0.110 & 0.094 & 0.094 & 0.187 \\
        iqr     & 0.319 & 0.267 & 0.275 & 0.486 & 0.149 & 0.125 & 0.121 & 0.238 \\
        \midrule
        \multicolumn{9}{c}{\( C_N = \log(N^{1/2}) \)} \\
        mean    & 1.023 & 1.016 & 1.018 & 1.957 & 1.004 & 1.002 & 1.002 & 1.038 \\
        median  & 1.013 & 1.002 & 1.005 & 1.028 & 1.004 & 0.998 & 0.998 & 1.006 \\
        std     & 0.292 & 0.246 & 0.247 & 2.685 & 0.135 & 0.113 & 0.114 & 0.419 \\
        iqr     & 0.363 & 0.307 & 0.325 & 0.672 & 0.181 & 0.155 & 0.153 & 0.307 \\
        \midrule
        \multicolumn{9}{c}{\( C_N = \log(N) \)} \\
        mean    & 1.217 & 1.181 & 1.184 & 4.163 & 1.019 & 1.011 & 1.011 & 1.608 \\
        median  & 1.061 & 1.027 & 1.032 & 1.301 & 1.007 & 0.999 & 1.001 & 1.039 \\
        std     & 1.157 & 1.139 & 1.153 & 4.235 & 0.236 & 0.186 & 0.185 & 0.214 \\
        iqr     & 0.629 & 0.464 & 0.461 & 8.604 & 0.305 & 0.241 & 0.241 & 0.561 \\
        \bottomrule
    \end{tabularx}
    \caption{Simulation results: $\lambda_1 = 0$, $\lambda_{2,0} = 0.5$ and $\lambda_{2,1} = 1.5$. Binary (1) and Binary (2) refer to standard binary choice models estimated using the tetrad-differencing technique, with cutoff thresholds $m_{\text{cutoff}}=1$ and $m_{\text{cutoff}}=2$, respectively.}
    \label{tab:simulation_results_0.5_1.5}
\end{table}
\FloatBarrier

\begin{figure}[htbp]
\centering
  \begin{subfigure}[t]{0.48\textwidth}
    \includegraphics[width=\textwidth]{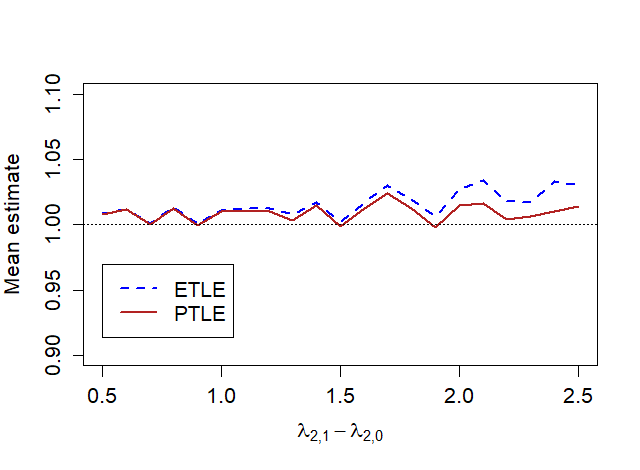}
    \caption{\(C_N = 0\)}
    \label{fig:cmle_0_heterogeneous}
  \end{subfigure}
  \hfill
  \begin{subfigure}[t]{0.48\textwidth}
    \includegraphics[width=\textwidth]{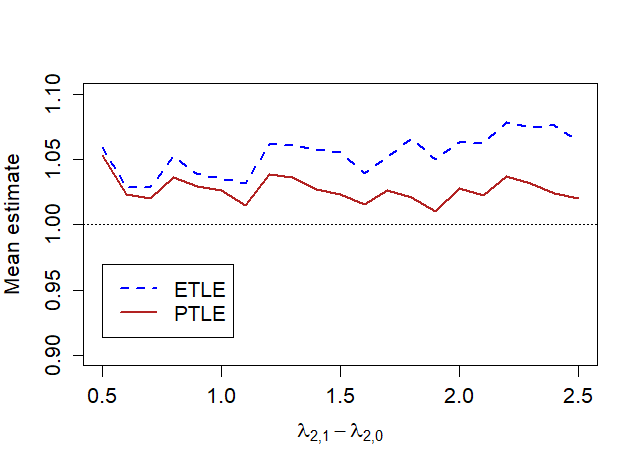}
    \caption{\(C_N = \log(N^{1/2})\)}
    \label{fig:cmle_logsqrt_heterogeneous}
  \end{subfigure}
  \caption{Mean PTLE and ETLE estimates by maximum threshold $\lambda_{2,1}$.}
  \label{fig:cmle_diff_all}
\end{figure}
\FloatBarrier

\subsection{Standard errors}

We compare the robust standard errors (RSE) calculated by our proposed sandwich variance estimator to the naive standard errors computed from the \texttt{glm()} function in R, which assumes independent observations and does not account for dyadic dependence in the network.
Table \ref{tab:simulation_results_se} presents the simulation results evaluating standard error performance for $N = 50$.
For each design, we report the ratio of the average estimated standard error to the Monte Carlo standard deviation, as well as empirical coverage rates at the 90\% and 95\% levels.

The RSE consistently achieves ratios close to 1 across all sparsity levels and threshold configurations, indicating accurate standard error estimation.
In contrast, the naive standard errors substantially underestimate variability, which leads to over-rejection of null hypotheses and excessive false positive rates.

The coverage rates for RSE remain stable and close to the nominal levels of 0.90 and 0.95 across all configurations.
The RSE correctly accounts for the sparsity and dependence structure inherent in the data.
In contrast, the naive standard errors yield severely undercovered intervals in every scenario.
In denser networks there are more informative tetrads and consequently more accumulated dependence in the conditional likelihood function. In sparse networks, on the other hand, there will be fewer informative tetrads and less accumulated dependence. Consequently, the naive standard errors perform even worse in denser network settings such as $C_N = 0$:
the actual coverage rates are only 0.090 and 0.107 for the 90\% and 95\% intervals, respectively.
These results confirm that naive standard errors systematically understate uncertainty and lead to overly narrow confidence intervals, regardless of sample size or network density.

Comparing across threshold configurations, RSE delivers reliable standard errors across all designs.
The naive standard errors remain substantially biased in both settings.
These findings demonstrate that researchers working with network models must account for the inherent dependence structure when conducting inference.
Standard statistical software that assumes independent observations will systematically understate uncertainty, leading to invalid conclusions.
Our robust variance estimator provides a practical solution that maintains reliable inference properties even under network sparsity and outcome category heterogeneity, making it essential for valid statistical inference in ordered network models.

\begin{table}[ht!]
\centering
\begin{tabularx}{\textwidth}{l *{4}{>{\centering\arraybackslash}X}}
    \toprule
    & \multicolumn{2}{c}{\( \lambda = \{0,2\} \)} & \multicolumn{2}{c}{\( \lambda = \{0,0.5\} \text{ and } \lambda = \{0,1.5\} \)} \\
    \cmidrule(lr){2-3} \cmidrule(lr){4-5}
    & NSE & RSE & NSE & RSE \\
    \midrule
    \multicolumn{5}{c}{\( C_N = 0 \)} \\
    se/std    & 0.067  & 1.050 & 0.062 & 1.034 \\
    90\% coverage     & 0.090 & 0.915 & 0.095 & 0.923 \\
    95\% coverage     & 0.107 & 0.964 & 0.107 & 0.963 \\
    \midrule
    \multicolumn{5}{c}{\( C_N = \log(\log(N)) \)} \\
    se/std    & 0.076 & 1.064 & 0.073 & 1.078 \\
    90\% coverage     & 0.099  & 0.918 & 0.110 & 0.927 \\
    95\% coverage     & 0.115 & 0.961 & 0.123 & 0.965 \\
    \midrule
    \multicolumn{5}{c}{\( C_N = \log(N^{1/2}) \)} \\
    se/std    & 0.086  & 1.053 & 0.083 & 1.041 \\
    90\% coverage     & 0.108 & 0.920 & 0.119 & 0.917 \\
    95\% coverage     & 0.125 & 0.960 & 0.137 & 0.956 \\
    \midrule
    \multicolumn{5}{c}{\( C_N = \log(N) \)} \\
    se/std    & 0.473  & 1.090  & 0.129 & 1.087 \\
    90\% coverage     & 0.171 & 0.925 & 0.200 & 0.935 \\
    95\% coverage     & 0.212 & 0.960 & 0.239 & 0.969 \\
    \bottomrule
\end{tabularx}
\caption{Performance of Naive and Robust Standard Errors ($N=50$)}
\label{tab:simulation_results_se}
\vspace{0.5em}
\footnotesize{\textit{Notes:} The table reports simulation results for two threshold configurations. We report the ratio of the average estimated standard error to the Monte Carlo standard deviation of the estimator (se/std) and empirical coverage rates for 90\% and 95\% confidence intervals. NSE refers to naive standard errors from a standard logit estimator assuming independent observations. RSE refers to the robust sandwich standard errors calculated using the method proposed in Section 4, which accounts for network dependence.}
\end{table}
\FloatBarrier

\section{Empirical illustration}
\label{sec:empirical-application}

\subsection{Data}

We analyze friendship formation among students using longitudinal relational measurements and nodal characteristics from a dataset of Dutch college students collected by \textcite{vanduijnFrameworkComparisonMaximum2009}.
The data track the same group of university freshmen over seven time points.
At the time of the first measurement, most students were strangers.
The first four waves of data were collected every three weeks, and the last three waves were collected every six weeks.
The dataset consists of 32 students.
Each student rated their relationship with others on a six-point ordinal scale (Table \ref{tab:relationship_categories}).
We focus on the fifth time point, when students had experienced 15 weeks of social interaction.

At this stage, most reported relationships fall into the categories of ``unknown person'' or ``neutral relationship'' (Table \ref{tab:category_distribution}).
We treat observations recorded as ``unknown person'' as missing data.
The ordinal position of this category is ambiguous, which makes its inclusion in the estimation problematic.
The corresponding degree distributions are summarized in Table \ref{tab:degree_distribution_data}.
Fewer links are observed when higher relationship levels are used as cutoff categories.
This sparsity poses a particular challenge for the ETLE, which can become unstable when the number of informative tetrads is small.

\begin{table}[ht]
\centering
\begin{tabularx}{\textwidth}{c l X}
\toprule
\textbf{Category} & \textbf{Label} & \textbf{Description of the response categories} \\
\midrule
0 & Troubled relationship & Persons with whom you cannot get on very well, and with whom you definitely do not want to start a relationship. There is a certain risk of getting into a conflict. \\[0.5em]
1 & Unknown person & Persons whom you do not know. \\[0.5em]
2 & Neutral relationship & Persons with whom you have not much in common. In case of an accidental meeting the contact is good. The chance of it growing into a friendship is not large. \\[0.5em]
3 & Friendly relationship & Persons with whom you regularly have pleasant contact during classes. The contact could grow into a friendship. \\[0.5em]
4 & Friendship & Persons with whom you have a good relationship, but whom you do not (yet) consider a `real' friend. \\[0.5em]
5 & Best friendship & Persons whom you would call your ```real'' friends. \\
\bottomrule
\end{tabularx}
\vspace{0.5em}
\caption{Description of relationship categories.}
\label{tab:relationship_categories}
\end{table}
\FloatBarrier

\begin{table}[ht]
\centering
\begin{tabular}{ccc}
\toprule
\textbf{Category} & \textbf{Frequency} & \textbf{Proportion} \\
\midrule
0 & 12  & 0.024 \\
 2 & 311 & 0.624 \\
 3 & 113 & 0.227 \\
 4 &  53 & 0.106 \\
 5 &   9 & 0.018 \\
\bottomrule
\end{tabular}
\caption{Category frequencies and proportions.}
\label{tab:category_distribution}
\end{table}
\FloatBarrier

\begin{table}[!htbp]
\centering
\begin{tabularx}{\textwidth}{l *{6}{>{\centering\arraybackslash}X}}
\hline
$m$ & $E \left[ 1\{Y_{ij} \geq m \} \right]$ & $Q_{0.25}$ & $Q_{0.5}$ & $Q_{0.75}$ & Min. & Max. \\
\hline
2 & 0.973 & 1.000 & 1.000 & 1.000 & 0.737 & 1.000 \\
3 & 0.375 & 0.223 & 0.344 & 0.492 & 0.000 & 1.000 \\
4 & 0.139 & 0.036 & 0.111 & 0.211 & 0.000 & 0.600 \\
5 & 0.020 & 0.000 & 0.000 & 0.000 & 0.000 & 0.200 \\
\hline
\end{tabularx}
\caption{Degree distribution of the friendship network by cutoff level.}
\label{tab:degree_distribution_data}
\vspace{0.5em}
\footnotesize{\textit{Notes:} The table summarizes the network's degree distribution when a link is defined as $Y_{ij} \geq m$. The ``Mean'' column reports the sample proportion of dyads meeting this criterion. $Q_{0.25}$, $Q_{0.5}$, and $Q_{0.75}$ are the 25th, 50th, and 75th percentiles of the node-level degree distribution. ``Min.'' and ``Max.'' are the minimum and maximum observed degree shares. The table shows that the network becomes increasingly sparse as the cutoff $m$ for defining a friendship tie increases.}
\end{table}
\FloatBarrier

The dataset contains individual characteristics including gender, smoking behavior, and education program.
We construct three dyad-level regressors: \textit{common gender}, \textit{both smoker}, and \textit{common program}.
The variable \textit{common gender} equals one when students $i$ and $j$ share the same gender and zero otherwise.
Similarly, \textit{both smokers} indicates whether both individuals in the dyad are smokers.
The variable \textit{common program} captures whether the two students are enrolled in the same education program.
All three variables are binary and symmetric in $(i,j)$.
Table \ref{tab:summary_stat_data} presents descriptive statistics.
Most pairs share the same gender (63.7\%).
Approximately 44\% of dyads are enrolled in the same education program, while only 11.6\% consist of two smokers.

\begin{table}[!htbp]
\centering
\begin{tabularx}{\textwidth}{l >{\centering\arraybackslash}X >{\centering\arraybackslash}X}
\hline
 & \text{Mean} & \text{Standard deviation} \\
\hline
Common gender  & 0.637 & 0.481 \\
Both smokers  & 0.116 & 0.321 \\
Common program & 0.440 & 0.497 \\
\hline
\end{tabularx}
\caption{Descriptive statistics for friendship data.}
\label{tab:summary_stat_data}
\end{table}
\FloatBarrier

\subsection{Results}
Table \ref{tab:friendship_estimates_5} presents estimation results for friendship outcomes at the fifth time point using various modeling approaches.
The ordered logit model without fixed effects shows a positive and significant coefficient for \textit{both smokers}, \textit{common gender}, and \textit{common program}, consistent with positive homophily.
However, this model may suffer from omitted variable bias due to unobserved heterogeneity.
The ordered logit model with fixed effects produces positive and highly significant coefficients for all three homophily variables.
Among them, \textit{both smokers} has the largest effect (2.040), followed by \textit{common gender} (1.268) and \textit{common program} (0.818).
These estimates support strong homophily patterns but remain susceptible to the incidental parameter problem.

Both ETLE and PTLE yield positive effects for all three homophily factors.
We do not report standard errors for the ETLE estimates, as we have only investigated the PTLE standard errors in this paper.
Under PTLE, all coefficients are statistically significant at the 5\% level or better.
While PTLE and ETLE yield the same ranking of homophily factors—assigning the strongest effect to \textit{both smokers}, followed by \textit{common program} and \textit{common gender}—the magnitudes of the estimates differ substantially between the two methods.
This divergence reflects the finite-sample sensitivity of ETLE, especially when high-score dyads are rare.

Binary models using single cutoffs yield positive homophily estimates with magnitudes varying by cutoff choice.
Both binary models assign the largest effect to \textit{both smokers}.
The $m=2$ model assigns the second largest effect to \textit{common program} (0.854), while the $m=3$ model amplifies the effect of \textit{common gender} (2.396) and diminishes the role of \textit{common program}, which becomes negative and insignificant.
These variations demonstrate the sensitivity of binary models to threshold selection.

\begin{table}
\centering 
\begin{tabular}{@{\extracolsep{5pt}}lcccccc} 
\toprule
 & \shortstack{Ordered\\logit w/o FE} 
 & \shortstack{Ordered\\logit w/ FE} 
 & ETLE 
 & PTLE 
 & \shortstack{Binary\\(2)} 
 & \shortstack{Binary\\(3)} \\ 
\midrule
Common gender & 0.558$^{**}$ & 1.268$^{***}$ & 0.186 & 0.723$^{**}$ & 0.554$^{*}$ & 2.396$^{***}$ \\
             & (0.203) & (0.269) &  & (0.286) & (0.299) & (0.812) \\

Both smokers & 1.022$^{***}$ & 2.040$^{***}$ & 3.472 & 2.431$^{***}$ & 2.397$^{***}$ & 2.805$^{*}$ \\
             & (0.266) & (0.476) &  & (0.765) & (0.779) & (1.490) \\

Common program & 0.473$^{*}$ & 0.818$^{***}$ & 0.544 & 0.732$^{**}$ & 0.854$^{**}$ & -0.452 \\
               & (0.193) & (0.224) & & (0.292) & (0.309) & (0.611) \\

\midrule
Observations & 498 & 498 & 3,319 & 3,319 & 2,649 & 616 \\
\bottomrule
\end{tabular}
\caption{Friendship estimates: Binary (2) and Binary (3) refer to standard binary choice models estimated using the tetrad-differencing technique, with cutoff thresholds $m_{\text{cutoff}}=2$ and $m_{\text{cutoff}}=3$, respectively. Significance levels: $^{*}p<0.10$, $^{**}p<0.05$, $^{***}p<0.01$.}
\label{tab:friendship_estimates_5}
\end{table}
\FloatBarrier

\subsection{Results from multiple time points}

To examine the temporal consistency of homophily patterns, we extend the analysis to all seven time points and plot the estimated coefficients for each dyad-level regressor across different estimation methods (Figure \ref{fig:coefficients}).
For each covariate, the top panel of each figure presents all six methods, while the bottom panel omits the most unstable one, Binary (3), to better highlight trends among the remaining approaches.
The plots enable a direct visual comparison of PTLE against other methods over time.

For the \textit{common gender} variable (Figure \ref{fig:common_gender}), PTLE delivers stable estimates across all time points, with moderate fluctuations around a central value.
ETLE broadly follows the pattern of PTLE but exhibits greater variability, particularly at early and late time points, due to the rarity of observations in some categories, such as ``best friendship.''
In contrast, Binary (3) produces an extremely large spike at time 3 and again at time 4, deviating substantially from the other methods.
These distortions likely reflect the method's sensitivity to sparsity in high-score friendship categories.
The standard ordered logit models, while not directly comparable in scale, produce a temporal trend for this coefficient that mirrors the one from PTLE.
Although the fixed effects specification aims to control for unobserved heterogeneity, it remains vulnerable to the incidental parameter problem in small samples, which can lead to biased estimates.

Turning to \textit{both smokers} (Figure \ref{fig:common_smoker}), PTLE again provides a coherent trend over time, with estimates gradually increasing and peaking around time 5 before declining slightly.
In contrast, ETLE produces substantially larger estimates at several time points, reflecting instability in sparse network settings.
Binary (3) exhibits highly erratic behavior, including an extreme estimate exceeding 9 at time 2, making interpretation unreliable.
PTLE avoids such volatility while still capturing meaningful temporal dynamics.
Once again, the standard ordered logit specifications produce a temporal trend similar to PTLE, though their estimates remain subject to the incidental parameter bias discussed previously.

The results for the \textit{common program} variable (Figure \ref{fig:common_program}) follow a similar pattern.
Binary (3) shows a sharp negative estimate at time 3 and high volatility overall.
ETLE displays an implausible peak at time 4, followed by inconsistent shifts.
Ordered logit models show similar patterns to PTLE, but the estimates differ in magnitude.
PTLE, by contrast, exhibits a gradual and interpretable evolution of the coefficient, increasing in later periods as program-based friendships become more prominent.
This trend aligns with expectations, as students have more time to interact and bond within their academic programs.

To further examine the relative strength of homophily factors, Figure \ref{fig:relative_importance} plots the estimated coefficients of \textit{both smokers} (top panel) and \textit{common program} (bottom panel), each expressed relative to the coefficient for \textit{common gender}.
Across most methods, the relative importance remains broadly consistent over time, suggesting a common relative ranking of homophily effects.
However, ETLE stands out as an exception, exhibiting substantial deviations at several time points.
These instabilities underscore ETLE's sensitivity to sparsity.

Our analysis of friendship among these Dutch students reveals strong and persistent homophily.
We find that students are significantly more likely to form stronger ties with others of the same gender, who share their smoking habits, or who are in the same academic program.
The PTLE estimator produces stable and intuitive results over time—for instance, capturing the growing importance of a shared program.
This stands in sharp contrast to the erratic results from other methods, demonstrating that our approach is essential for drawing reliable conclusions about the drivers of network formation.

\begin{figure}[htbp]
\centering

\begin{subfigure}[b]{0.45\textwidth}
    \centering
    \includegraphics[width=\textwidth]{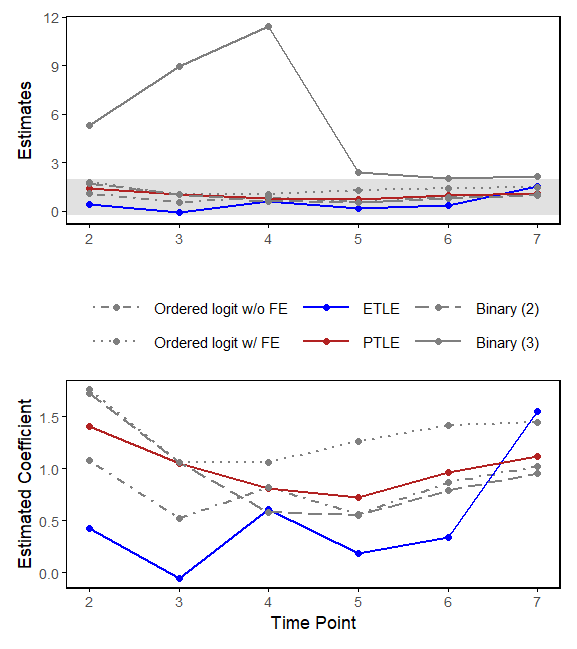}
    \caption{Common Gender}
    \label{fig:common_gender}
\end{subfigure}
\hfill
\begin{subfigure}[b]{0.45\textwidth}
    \centering
    \includegraphics[width=\textwidth]{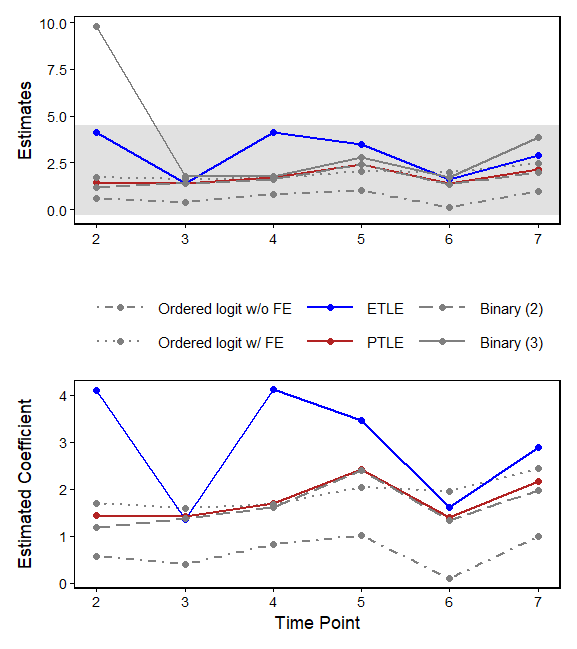}
    \caption{Both Smokers}
    \label{fig:common_smoker}
\end{subfigure}

\vspace{0.5cm}

\begin{subfigure}[b]{0.45\textwidth}
    \centering
    \includegraphics[width=\textwidth]{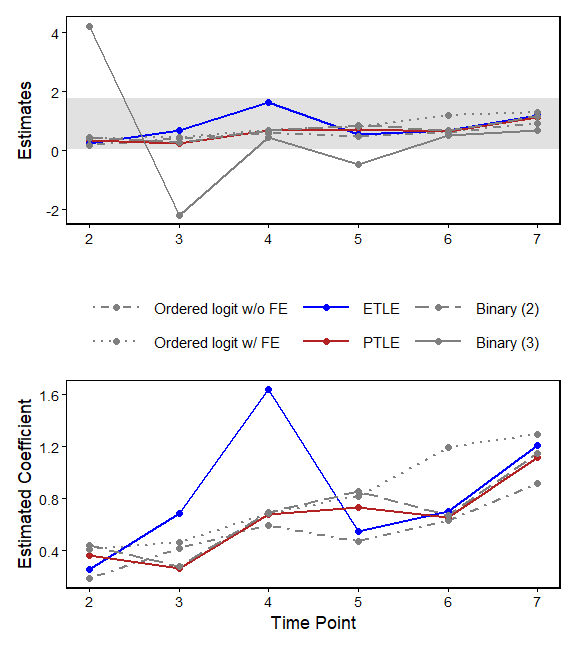}
    \caption{Common Program}
    \label{fig:common_program}
\end{subfigure}
\hfill
\begin{subfigure}[b]{0.45\textwidth}
    \centering
    \includegraphics[width=\textwidth]{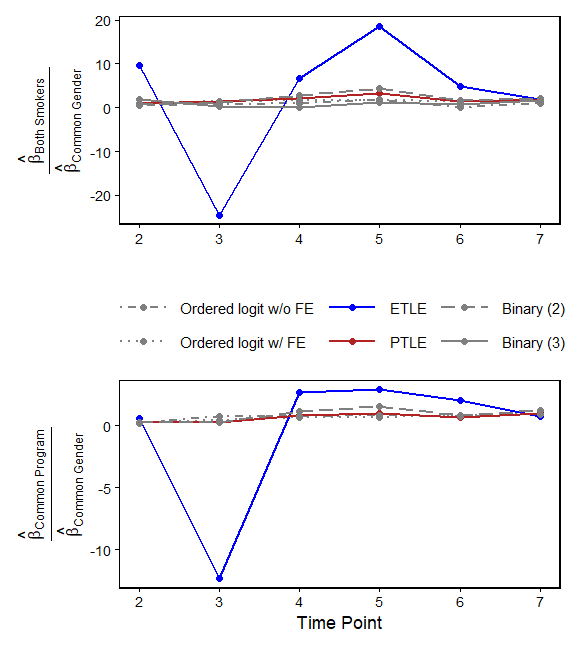}
    \caption{Relative importance}
    \label{fig:relative_importance}
\end{subfigure}

\caption{Estimated coefficients across time points for (a) \textit{common gender}, (b) \textit{both smokers}, (c) \textit{common program}, and (d) relative importance.}
\label{fig:coefficients}
\end{figure}
\FloatBarrier

\section{Alternative fixed effects specifications}
\label{sec:alternative_specifications}

We consider two more restrictive models.

\subsection{Additive fixed effects and common thresholds}
\label{sec:modelI}

In this specification, we impose two key restrictions on the general threshold structure $\lambda_{ijm}^*$ in \eqref{eq:threshold}.
First, we assume common baseline thresholds $\lambda_m$ for all dyads $(i,j)$.
Second, we restricts the node-specific heterogeneity to be additively separable and constant across outcome categories, represented by a sender effect $\alpha_i$ and a receiver effect $\gamma_j$, such that the overall threshold is $\lambda_{ijm}^* = \lambda_m - \alpha_i - \gamma_j$.
This structure is more restrictive than our main model, which allowed sender ($\lambda_{im}$) and receiver ($\delta_{jm}$) effects to vary freely across categories $m$.
In a friendship context, this simpler specification constrains individuals' tendencies to form relationships to be uniform across all rating levels, unlike our main model which permits different thresholds at different points in the relationship scale.
This additive formulation parallels the two-way fixed effects commonly used in panel data ordered choice models (\textcite{dasPanelDataModel1999, johnsonPanelDataModels2004, baetschmannIdentificationEstimationThresholds2012, Baetschmann2015, murisEstimationFixedEffectsOrdered2017, botosaruIdentificationTimevaryingTransformation2023}), where sender/receiver effects correspond to individual/time effects in that setting.

In this model, the thresholds
$$\lambda_{ijm}^* = \lambda_{m} - \alpha_i - \gamma_j$$
consist of a common threshold $\lambda_m$, and fixed effects $(\alpha_i,\gamma_j)$ for outgoing and receiving nodes.
Node-specific fixed effects are gathered in
$F_{i} = (\alpha_{i}, \gamma_{i})$ and
$\mathbf F = (F_{i}: \; i \in \mathbb N_N)$, 
while $\mathbf Y$ and $\mathbf X$ are as before.
In the following model specification,
$\beta_0$ is the true value of the regression coefficient $\beta$,
and $\lambda_{m0}$ the true value of the $m$th threshold, $m=1,\cdots,M$.
Gather the common parameters in
$\theta = (\beta, \lambda_1, \cdots, \lambda_{M})$
with true value
$\theta_0$.

\begin{assumption}\label{a:likelihood_I}
    The conditional likelihood of $\mathbf Y = \mathbf y$ given the covariates $\mathbf X$ and the fixed effects $\mathbf F = (F_i : i \in \mathbb N_N)$, where $F_i = (\alpha_i, \gamma_i)$ for this specification, is
    $$P(\left. \mathbf Y = \mathbf y \right| \mathbf X, \mathbf{F}) =
        \prod_{(i,j)\in\mathcal I_N} P(Y_{ij} = y_{ij}|X_{ij}, F_i, F_j),$$
    with the probability for a single dyad given by
    $$
    P(Y_{ij} = m | X_{ij}, F_i, F_j) =
    \begin{cases}
    1 - \Lambda\left(X_{ij}^\prime \beta_0 + \alpha_{i} + \gamma_{j} - \lambda_{10}\right), & m = 0, \\
    \Lambda\left(X_{ij}^\prime \beta_0 + \alpha_{i} + \gamma_{j} - \lambda_{m0}\right) \\
    \phantom{--} - \Lambda\left(X_{ij}^\prime \beta_0 + \alpha_{i} + \gamma_{j} - \lambda_{m+1,0}\right), & 1 \leq m \leq M-1, \\
    \Lambda\left(X_{ij}^\prime \beta_0 + \alpha_{i} + \gamma_{j} - \lambda_{M0}\right), & m = M.
    \end{cases}
    $$
\end{assumption}

To identify parameters without estimating all fixed effects, we transform our ordered model into a set of binary models. 
For each threshold $m$, we define binary indicators $D_{ij}(m) = \mathbf{1}\{Y_{ij} \geq m\}$. 
Under Assumption \ref{a:likelihood_I}, the probability that this indicator equals one is 
\begin{align*}
P(D_{ij}(m) = 1 | X_{ij}, F_i, F_j) &=
P(Y_{ij} \geq m | X_{ij}, F_i, F_j) \\
&= \Lambda(X_{ij}^\prime \beta_0 + \alpha_i + \gamma_j - \lambda_{m0}).
\end{align*}
By defining $\widetilde \alpha_i = \alpha_i - \lambda_m$, we can rewrite this as $\Lambda(X_{ij}^\prime \beta_0 + \widetilde \alpha_i + \gamma_j)$, which shows that we cannot separately identify the level of the fixed effects and the threshold parameter from a single cutoff $m$.

For a fixed value of $m$, the variables $D_{ij}(m)$ follow the structure of a standard binary choice model for dyadic data with additive sender and receiver fixed effects.
Applying tetrad differencing techniques allows for the identification and estimation of $\beta_0$ by eliminating the additive fixed effects $\tilde{\alpha}_i$ and $\gamma_j$.
Furthermore, by using different combinations of thresholds $\mathbf{m}$, we can also identify differences in the common thresholds, $\lambda_{m0} - \lambda_{m'0}$, as shown in related panel data contexts (\textcite{baetschmannIdentificationEstimationThresholds2012, murisEstimationFixedEffectsOrdered2017}).

To that end, let $\mathbf m = (m_{11},m_{12},m_{21},m_{22})$ denote a vector of four potentially different cutoffs associated with the dyadic outcomes within a tetrad.
The set of all possible tetrads $\Sigma$ is as defined in Equation \eqref{def:Sigma}.

Recall the definition of the tetrad statistic $Z_\sigma(\mathbf m)$ from Equation \eqref{def:Z}.
This statistic depends on the binary outcomes $D_{i_r,j_s}(m_{rs})$ derived from potentially different cutoffs within the tetrad.
As we will show, its conditional distribution given the covariates $X_\sigma$ depends on the differenced covariate vector $r_\sigma = (X_{i_1,j_1} - X_{i_1,j_2}) - (X_{i_2,j_1} - X_{i_2,j_2})$ and $\lambda_0(\mathbf m)$.

\begin{thm}\label{thm:directed_I_sufficiency}
    For the model in Assumption \ref{a:likelihood_I},
    $$
    P\left(\left. Z_{\sigma}(\mathbf m) = 1 \right| Z_{\sigma}(\mathbf m) \in \{-1,+1\}, X_\sigma\right) = \Lambda\left(r_\sigma^\prime \beta_0 - \lambda_0(\mathbf m)\right),
    $$
    where $$
    \lambda_0(\mathbf m) =
    (\lambda_{m_{11}0} - \lambda_{m_{12}0})
    -
    (\lambda_{m_{21}0} - \lambda_{m_{22}0}).
$$
\end{thm}
\begin{proof}See Appendix \ref{sec:proofs}.\end{proof}

Theorem \ref{thm:directed_I_sufficiency} provides the foundation for identifying the parameter vector $\theta_0$ from the binary choice probabilities implied by Assumption \ref{a:likelihood_I}.
This approach is related to identification strategies used in other fixed-effects ordered choice models (\textcite{murisEstimationFixedEffectsOrdered2017, abrevayaIntervalCensoredRegression2020, botosaruIdentificationTimevaryingTransformation2023}).

First, for any $m \in \{1,\cdots,M\}$, apply Theorem \ref{thm:directed_I_sufficiency} with the common cutoff vector $\mathbf m = (m,m,m,m)$.
In this case, the threshold difference term $\lambda_0(\mathbf m) = (\lambda_{m0} - \lambda_{m0}) - (\lambda_{m0} - \lambda_{m0}) = 0$.
Assuming standard rank conditions, Theorem \ref{thm:directed_I_sufficiency} implies
$$\beta_0 = \left(E[r_\sigma r_\sigma^\prime]\right)^{-1}
                  E[r_\sigma^\prime \Lambda^{-1}(P\left(\left. Z_\sigma(\mathbf m) = 1 \right| Z_\sigma(\mathbf m) \in \{-1,+1\}, X_\sigma\right))].$$

Second, with $\beta_0$ identified from the first step, we can identify the threshold parameters relative to a normalization.
Let $\lambda_{10} = 0$.
Now apply Theorem \ref{thm:directed_I_sufficiency} with the cutoff vector $\mathbf{m} = (m, 1, 1, 1)$ for any $m \in \{2, \dots, M\}$.
For this choice, the threshold difference term is $\lambda_0((m,1,1,1)) = (\lambda_{m0} - \lambda_{10}) - (\lambda_{10} - \lambda_{10}) = \lambda_{m0}$ (using the normalization).
Theorem \ref{thm:directed_I_sufficiency} gives the identification result for the $m$-th threshold:
\begin{align*}
\lambda_{m0} &= - E\left[\Lambda^{-1}
(P\left(\left. Z_\sigma(\mathbf m) = 1 \right| Z_\sigma(\mathbf m) \in \{-1,+1\}, X_\sigma\right)) - r_\sigma^\prime \beta_0\right].
\end{align*}

This analysis suggests that the conditional probabilities $P(Z_\sigma(\mathbf m) = 1 | Z_\sigma(\mathbf m) \in \{-1,+1\}, X_\sigma)$ can be used in a conditional likelihood approach for estimation.
For any given choice of cutoff vector $\mathbf m = (m_{11}, m_{12}, m_{21}, m_{22})$, Theorem \ref{thm:directed_I_sufficiency} provides a conditional likelihood contribution based on the probability $\Lambda(r_\sigma^\prime \beta - \lambda_0(\mathbf m))$, allowing estimation of $\beta_0$ and the specific threshold difference $\lambda_0(\mathbf m)$ from informative tetrads associated with this specific $\mathbf{m}$.

To estimate the full parameter vector $\theta_0 = (\beta_0, \lambda_{10}, \dots, \lambda_{M0})$ efficiently using all available information, we define an estimator analogous to the PTLE proposed for the main model.
Let $\mathcal{M} = \{1, \dots, M\}^4$ be the set of all possible cutoff vectors.
Let $S_{\sigma \mathbf{m}} = \mathbf{1}\{Z_\sigma(\mathbf{m}) \in \{-1, 1\}\}$ indicate if tetrad $\sigma$ is informative for vector $\mathbf{m}$, and let $y^*_{\sigma \mathbf{m}} = \mathbf{1}\{Z_\sigma(\mathbf{m}) = 1\}$ be the corresponding binary outcome.
The pooled CML objective function is:
\begin{equation}
    \label{eq:pooled_cml_model_a}
    L_N(\theta) = \sum_{\sigma \in \Sigma} \sum_{\mathbf{m} \in \mathcal{M}} S_{\sigma \mathbf{m}} \left[ y^*_{\sigma \mathbf{m}} \log \Lambda_{\sigma \mathbf{m}}(\theta) + (1 - y^*_{\sigma \mathbf{m}}) \log(1 - \Lambda_{\sigma \mathbf{m}}(\theta)) \right],
\end{equation}
where $\theta = (\beta, \lambda_2, \dots, \lambda_M)$ incorporates the normalization $\lambda_1=0$, and $\Lambda_{\sigma \mathbf{m}}(\theta) = \Lambda(r_\sigma^\prime \beta - \lambda_0(\mathbf{m}; \theta))$ with $\lambda_0(\mathbf{m}; \theta) = (\lambda_{m_{11}} - \lambda_{m_{12}}) - (\lambda_{m_{21}} - \lambda_{m_{22}})$ being evaluated using the parameters in $\theta$ (and $\lambda_1=0$).
The pooled CML estimator $\hat{\theta}_N$ maximizes $L_N(\theta)$.

\subsection{Heterogeneous sender thresholds and common receiver effects}
\label{sec:modelIII}

This specification offers an intermediate structure between the additive model (Assumption \ref{a:likelihood_I}) and our main model (Assumption \ref{a:likelihood_II}).
Here, we assume thresholds $\lambda_{ijm}^{*} = \lambda_{im} - \gamma_j$, allowing sender-specific components $\lambda_{im}$ to vary freely across categories $m$ while constraining receiver effects $\gamma_j$ to be constant across categories.
In a friendship network context, this permits individuals to have different threshold profiles when rating others (e.g., distinguishing between ``neutral'' and ``best friend'' relationships with varying strictness), while restricting the effect of being rated by others to a simple constant shift.
This contrasts with our main model where both sender and receiver effects could vary by category, and with the additive model where neither could vary by category.

Data $(\mathbf X, \mathbf Y)$ are defined as before.
The incidental parameters for this specification are denoted by $F_i = ((\gamma_i,\lambda_{im}), m = 1, \cdots, M)$.
As before, $\lambda_{im}$ are category-specific threshold shifters that apply when $i$ is the sender node.
The parameter $\gamma_i$ is a node-specific fixed effect applied when $i$ is the receiver node, which is restricted to be constant across categories $m$.
We gather the incidental parameters across all nodes into the vector $\mathbf F = (F_i : i \in \mathbb N_N)$.

\begin{assumption}\label{a:likelihood_III}
    The conditional likelihood of the ordered choice $\mathbf Y = \mathbf y$ given the covariates $\mathbf X$ and the fixed effects $\mathbf{F}$ is
    $$P(\mathbf Y = \mathbf y| \mathbf X, \mathbf{F}) =
        \prod_{(i,j)\in\mathcal I_N} P(Y_{ij} = y_{ij}|X_{ij}, F_i, F_j),$$
    with the probability for a single dyad given by
    $$
    P(Y_{ij} = m | X_{ij}, F_i, F_j) =
    \begin{cases}
    1 - \Lambda\left(
    X_{ij}^\prime \beta_0 + \gamma_{j} - \lambda_{i1}\right), & m = 0, \\

    \Lambda\left(
    X_{ij}^\prime \beta_0 +\gamma_{j} - \lambda_{im}
    \right)
    -
    \Lambda\left(
    X_{ij}^\prime \beta_0 +\gamma_{j} - \lambda_{i,m+1}
    \right), & 1 \leq m \leq M-1, \\

    \Lambda\left(
    X_{ij}^\prime \beta_0 +\gamma_{j} - \lambda_{iM}
    \right), & m = M.
    \end{cases}
    $$
\end{assumption}

Under Assumption \ref{a:likelihood_III}, the probability for the binary variable $D_{ij}(m)=\mathbf{1}\{Y_{ij} \ge m\}$ is given by:
$$
P(D_{ij}(m) = 1 | X_{ij}, F_i, F_j) = \Lambda( X_{ij}^\prime \beta_0 +\gamma_{j} - \lambda_{im}).
$$

Our identification strategy here differs from that used for the main model (Theorem \ref{thm:directed_II_sufficiency}) due to the specific structure of the fixed effects ($\lambda_{im}, \gamma_j$) in Assumption \ref{a:likelihood_III}.
The distribution of the general tetrad difference $Z_\sigma(\mathbf{m})$ depends on the difference-in-differences of the dyad-specific thresholds $\lambda_{ijm}^* = \lambda_{im} - \gamma_j$.
This difference term, $\Delta\lambda(\mathbf{m}) = (\lambda^*_{i_1j_1m_{11}} - \lambda^*_{i_1j_2m_{12}}) - (\lambda^*_{i_2j_1m_{21}} - \lambda^*_{i_2j_2m_{22}})$, simplifies in this specification because the receiver effects $\gamma_j$ cancel out, leaving $\Delta\lambda(\mathbf{m}) = (\lambda_{i_1m_{11}} - \lambda_{i_1m_{12}}) - (\lambda_{i_2m_{21}} - \lambda_{i_2m_{22}})$.
For this term to be zero without imposing restrictions on the $\lambda_{im}$ values themselves, we require the terms within each parenthesis to vanish.
Thus, a sufficiency result yielding a conditional probability free of fixed effects under Assumption \ref{a:likelihood_III} requires $m_{11} = m_{12}$ and $m_{21} = m_{22}$.

Based on the requirement that $m_{11}=m_{12}$ and $m_{21}=m_{22}$ for the fixed effects to cancel, we define a specialized tetrad statistic for this specification.
Let $(m,m') \in \{1,\cdots,M\}^2$ be any pair of cutoffs.
For any tetrad $\sigma$, we define the difference-in-differences statistic $Z_\sigma(m;m')$ using the general definition in Equation \eqref{def:Z} with the specific cutoff vector $\mathbf{m}=(m,m,m',m')$:
$$ Z_\sigma(m;m') \equiv Z_\sigma((m,m,m',m')). $$
This statistic uses a common cutoff $m$ for dyads involving sender $i_1$ and potentially a different common cutoff $m'$ for dyads involving sender $i_2$.

\begin{thm}\label{thm:directed_III_sufficiency}
    For the model in Assumption \ref{a:likelihood_III}, and for any $(m,m') \in \{ 1,\cdots,M\}^2$,
    $$
    P\left(\left. Z_\sigma(m;m') = 1 \right| Z_\sigma(m;m') \in \{-1,+1\}, X_{\sigma}\right) = \Lambda\left(r_\sigma^\prime \beta_0 \right).
    $$
\end{thm}
\begin{proof}See Appendix \ref{sec:proofs}.\end{proof}

\appendix
\section{Proofs}\label{sec:proofs}
\subsection{Sufficiency}

\begin{proof}[Proof of Theorem \ref{thm:directed_II_sufficiency}]
\label{proof:sufficiency_II}
    The technique used here is similar to that in the proof of Theorem \ref{thm:directed_I_sufficiency}.
    We first demonstrate why using the same cutoff $m$ across all dyads in the tetrad is necessary for identification in this model.
    We start with the general case using potentially different cutoffs $\mathbf m = (m_{11},m_{12},m_{21},m_{22})$ and show that the fixed effects only cancel when $m_{11}=m_{12}=m_{21}=m_{22}=m$.

    Let $\sigma = (i_1, i_2, j_1, j_2)$ be a tetrad.
    Define the shorthand $\eta_{rs} = X_{i_r,j_s}^\prime \beta_0 - \lambda_{i_r,m_{rs}} - \delta_{j_s,m_{rs}}$ for $r,s \in \{1,2\}$.
    Define the normalization constant:
    \begin{align*}
        K = \prod_{r=1}^2 \prod_{s=1}^2 [1 + \exp(\eta_{rs})].
    \end{align*}
    This $K$ is the denominator term arising from the product of the four logistic probabilities.

    Under Assumption \ref{a:likelihood_II}, the conditional probability for the binarized variable $D_{ij}(m) = \mathbf{1}\{Y_{ij} \ge m\}$ is:
    $$P(D_{ij}(m) = 1 | X_{ij}, F_i, F_j) = \Lambda( X_{ij}^\prime \beta_0 - \lambda_{im} - \delta_{jm}).$$
    These variables $D_{ij}(m_{rs})$ are independent across dyads $(i,j)$ conditional on covariates $\mathbf X$ and fixed effects $\mathbf F$ (Assumption \ref{a:likelihood_II}).

    Recall the definition $Z_\sigma(\mathbf m) = \frac{1}{2}((D_{i_1,j_1}(m_{11})-D_{i_1,j_2}(m_{12})) - (D_{i_2,j_1}(m_{21})-D_{i_2,j_2}(m_{22})))$.
    Using the independence conditional on $X_\sigma$ and $\mathbf{F}$, we find:
    \begin{align*}
        P(Z_\sigma(\mathbf m)=+1 | X_\sigma, \mathbf{F})
        &= \Lambda(\eta_{11}) (1-\Lambda(\eta_{12})) (1-\Lambda(\eta_{21})) \Lambda(\eta_{22}) \\
        &= \frac{\exp(\eta_{11}) \times 1 \times 1 \times \exp(\eta_{22})}{K} = \frac{\exp(\eta_{11} + \eta_{22})}{K},
        \\
        P(Z_\sigma(\mathbf m)=-1 | X_\sigma, \mathbf{F})
        &= (1-\Lambda(\eta_{11})) \Lambda(\eta_{12}) \Lambda(\eta_{21}) (1-\Lambda(\eta_{22})) \\
        &= \frac{1 \times \exp(\eta_{12}) \times \exp(\eta_{21}) \times 1}{K} = \frac{\exp(\eta_{12} + \eta_{21})}{K}.
    \end{align*}

    The odds ratio is therefore:
    \begin{align*}
        \frac{P(Z_\sigma(\mathbf m)=+1| X_\sigma, \mathbf{F})}{P(Z_\sigma(\mathbf m)=-1| X_\sigma, \mathbf{F})}
        &= \frac{\exp(\eta_{11} + \eta_{22})}{\exp(\eta_{12} + \eta_{21})} \\
        &= \exp(\eta_{11} + \eta_{22} - \eta_{12} - \eta_{21}) \\
        &= \exp\left( (X_{i_1,j_1}^\prime + X_{i_2,j_2}^\prime - X_{i_1,j_2}^\prime - X_{i_2,j_1}^\prime) \beta_0 \right. \\
        & \qquad \left. - [ (\lambda_{i_1,m_{11}} + \delta_{j_1,m_{11}}) + (\lambda_{i_2,m_{22}}+ \delta_{j_2,m_{22}}) \right. \\
        & \qquad \qquad \left. - (\lambda_{i_1,m_{12}}+ \delta_{j_2,m_{12}}) - (\lambda_{i_2,m_{21}} + \delta_{j_1,m_{21}}) ] \right) \\
        &= \exp( r_\sigma^\prime \beta_0 - \Delta\lambda(\mathbf{m}) ),
    \end{align*}
    where $\Delta\lambda(\mathbf{m})$ is as in \eqref{def:DeltaLambda}.
    The conditional probability $P(Z_\sigma(\mathbf m)=1| Z_\sigma(\mathbf m) \in \{-1,1\}, X_\sigma, \mathbf{F})$ is $\Lambda(r_\sigma^\prime \beta_0 - \Delta\lambda(\mathbf{m}))$.

    For inference on $\beta_0$ free of $\mathbf{F}$, we need $\Delta\lambda(\mathbf{m})=0$.
    Given the structure $\Delta\lambda(\mathbf{m}) = (\lambda_{i_1,m_{11}} - \lambda_{i_1,m_{12}}) + \dots$, this term equals zero for arbitrary fixed effects if and only if the indices within each difference are equal.
    That is, we require $m_{11}=m_{12}$ (for $\lambda_{i_1}$ terms), $m_{22}=m_{21}$ (for $\lambda_{i_2}$ terms), $m_{11}=m_{21}$ (for $\delta_{j_1}$ terms), and $m_{22}=m_{12}$ (for $\delta_{j_2}$ terms).
    These conditions together imply $m_{11} = m_{12} = m_{21} = m_{22} = m$ for some $m \in \{1, \dots, M\}$.
    When $\mathbf m = (m,m,m,m)$, the statistic is $Z_\sigma(\mathbf m) = \overline Z_\sigma(m)$ (defined in Eq. \eqref{def:Z_bar}), and the fixed effects difference $\Delta\lambda((m,m,m,m))=0$.

    The odds ratio simplifies to:
    \begin{align*}
        \frac{P(\overline Z_\sigma(m)=+1| X_\sigma, \mathbf{F})}{P(\overline Z_\sigma(m)=-1| X_\sigma, \mathbf{F})}
        &= \exp( r_\sigma^\prime \beta_0 - 0 ) \\
        &= \exp(r_\sigma^\prime \beta_0).
    \end{align*}
    From this, the conditional probability $P(\overline{Z}_\sigma(m)=1 | \overline{Z}_\sigma(m) \in \{-1,1\}, X_\sigma, \mathbf{F})$ is:
    $$ \frac{\exp(r_\sigma^\prime \beta_0)}{1 + \exp(r_\sigma^\prime \beta_0)} = \Lambda(r_\sigma^\prime \beta_0). $$
    Since this final probability does not depend on $\mathbf{F}$, we can write it as $P(\overline{Z}_\sigma(m)=1 | \overline{Z}_\sigma(m) \in \{-1,1\}, X_\sigma)$, establishing Theorem \ref{thm:directed_II_sufficiency}.
\end{proof}

\begin{proof}[Proof of Theorem \ref{thm:directed_I_sufficiency}]
\label{proof:directed_I_sufficiency}
    Let $\sigma = (i_1, i_2, j_1, j_2)$ be a fixed tetrad and \[\mathbf{m} = (m_{11}, m_{12}, m_{21}, m_{22})\] be fixed cutoffs. 
    Define the shorthand $\zeta_{rs} = X_{i_r,j_s}^\prime \beta_0 + \alpha_{i_r} + \gamma_{j_s} - \lambda_{m_{rs}0}$ for $r,s \in \{1,2\}$.
    Define the normalization constant:
    \begin{align*}
        K &= \prod_{r=1}^2 \prod_{s=1}^2 [1 + \exp(\zeta_{rs})].
    \end{align*}
    This $K$ is the denominator term arising from the product of the four logistic probabilities.

    Under Assumption \ref{a:likelihood_I}, the conditional probability for the binarized variable $D_{ij}(m) = \mathbf{1}\{Y_{ij} \ge m\}$ is:
    $$P(D_{ij}(m) = 1 | X_{ij}, F_i, F_j) = \Lambda(X_{ij}^\prime \beta_0 + \alpha_i + \gamma_j  - \lambda_{m0}).$$
    These variables $D_{ij}(m_{rs})$ are independent across dyads $(i,j)$ conditional on covariates $\mathbf X$ and fixed effects $\mathbf F$ (Assumption \ref{a:likelihood_I}).

    Using the independence conditional on $X_\sigma$ and $\mathbf{F}$, we find the probabilities of informative tetrads:
    \begin{align*}
        P(Z_\sigma(\mathbf m)=+1 | X_\sigma, \mathbf{F})
        &= \Lambda(\zeta_{11}) (1-\Lambda(\zeta_{12})) (1-\Lambda(\zeta_{21})) \Lambda(\zeta_{22}) \\
        &= \frac{\exp(\zeta_{11}) \times 1 \times 1 \times \exp(\zeta_{22})}{K} = \frac{\exp(\zeta_{11} + \zeta_{22})}{K},
        \\
        P(Z_\sigma(\mathbf m)=-1 | X_\sigma, \mathbf{F})
        &= (1-\Lambda(\zeta_{11})) \Lambda(\zeta_{12}) \Lambda(\zeta_{21}) (1-\Lambda(\zeta_{22})) \\
        &= \frac{1 \times \exp(\zeta_{12}) \times \exp(\zeta_{21}) \times 1}{K} = \frac{\exp(\zeta_{12} + \zeta_{21})}{K}.
    \end{align*}

    The odds ratio is therefore:
    \begin{align*}
        \frac{P(Z_\sigma(\mathbf m)=+1| X_\sigma, \mathbf{F})}{P(Z_\sigma(\mathbf m)=-1| X_\sigma, \mathbf{F})}
        &= \frac{\exp(\zeta_{11} + \zeta_{22})}{\exp(\zeta_{12} + \zeta_{21})} \\
        &= \exp(\zeta_{11} + \zeta_{22} - \zeta_{12} - \zeta_{21}) \\
        &= \exp\left( (X_{i_1,j_1}^\prime + X_{i_2,j_2}^\prime - X_{i_1,j_2}^\prime - X_{i_2,j_1}^\prime) \beta_0 \right. \\
        & \qquad \left. + (\alpha_{i_1} + \gamma_{j_1} - \lambda_{m_{11}0}) + (\alpha_{i_2} + \gamma_{j_2} - \lambda_{m_{22}0}) \right. \\
        & \qquad \qquad \left. - (\alpha_{i_1} + \gamma_{j_2} - \lambda_{m_{12}0}) - (\alpha_{i_2} + \gamma_{j_1} - \lambda_{m_{21}0}) \right) \\
        &= \exp\left( r_\sigma^\prime \beta_0 + (\alpha_{i_1} - \alpha_{i_1}) + (\alpha_{i_2} - \alpha_{i_2}) + (\gamma_{j_1} - \gamma_{j_1}) + (\gamma_{j_2} - \gamma_{j_2}) \right. \\
        & \qquad \qquad \left. - [ (\lambda_{m_{11}0} - \lambda_{m_{12}0}) - (\lambda_{m_{21}0} - \lambda_{m_{22}0}) ] \right) \\
        &= \exp( r_\sigma^\prime \beta_0 - \lambda_0(\mathbf m) ),
    \end{align*}
    where $\lambda_0(\mathbf m)$ is as in Theorem \ref{thm:directed_I_sufficiency}. 

    The conditional probability $P(Z_\sigma(\mathbf m)=1| Z_\sigma(\mathbf m) \in \{-1,1\}, X_\sigma, \mathbf{F})$ is obtained from the odds ratio $O = \exp( r_\sigma^\prime \beta_0 - \lambda_0(\mathbf m) )$ as $O / (1+O)$:
    \begin{align*}
        P(Z_\sigma(\mathbf m)=1| Z_\sigma(\mathbf m) \in\{-1,1\}, X_\sigma, \mathbf{F})
        &= \frac{P(Z_\sigma(\mathbf m)=1| X_\sigma, \mathbf{F})}
             {P(Z_\sigma(\mathbf m)=1| X_\sigma, \mathbf{F}) + P(Z_\sigma(\mathbf m)=-1| X_\sigma, \mathbf{F})} \\
        &= \frac{P(Z_\sigma(\mathbf m)=1| X_\sigma, \mathbf{F}) / P(Z_\sigma(\mathbf m)=-1| X_\sigma, \mathbf{F})}
             {P(Z_\sigma(\mathbf m)=1| X_\sigma, \mathbf{F})/P(Z_\sigma(\mathbf m)=-1| X_\sigma, \mathbf{F}) + 1} \\
        &= \frac{\exp\left( r_\sigma^\prime \beta_0 - \lambda_0(\mathbf m)\right)}
             {\exp\left( r_\sigma^\prime \beta_0 - \lambda_0(\mathbf m)\right) + 1} \\
        &= \Lambda\left( r_\sigma^\prime \beta_0 - \lambda_0(\mathbf m)\right).
    \end{align*}
\end{proof}

\begin{proof}[Proof of Theorem \ref{thm:directed_III_sufficiency}]
\label{proof:directed_III_sufficiency}
    Let $\sigma = (i_1, i_2, j_1, j_2)$ be a fixed tetrad. 
    We use the specific cutoff vector $\mathbf{m}=(m,m,m',m')$ as required by the specialized tetrad statistic $Z_\sigma(m;m') \equiv Z_\sigma((m,m,m',m'))$ defined in Section \ref{sec:modelIII}.

    Under Assumption \ref{a:likelihood_III}, the threshold structure is $\lambda_{ijm}^* = \lambda_{im} - \gamma_j$. The probability for the binarized variable $D_{ij}(k) = \mathbf{1}\{Y_{ij} \ge k\}$ is
    $$ P(D_{ij}(k)=1 | X_{ij}, F_i, F_j) = \Lambda(X_{ij}'\beta_0 + \gamma_j - \lambda_{ik}). $$
    Define the shorthand $\xi_{rs} = X_{i_r,j_s}^\prime \beta_0 + \gamma_{j_s} - \lambda_{i_r, m_{rs}}$, where $m_{11}=m$, $m_{12}=m$, $m_{21}=m'$, and $m_{22}=m'$. Specifically:
    \begin{align*}
        \xi_{11} &= X_{i_1,j_1}'\beta_0 + \gamma_{j_1} - \lambda_{i_1, m} \\
        \xi_{12} &= X_{i_1,j_2}'\beta_0 + \gamma_{j_2} - \lambda_{i_1, m} \\
        \xi_{21} &= X_{i_2,j_1}'\beta_0 + \gamma_{j_1} - \lambda_{i_2, m'} \\
        \xi_{22} &= X_{i_2,j_2}'\beta_0 + \gamma_{j_2} - \lambda_{i_2, m'}
    \end{align*}
    Define the normalization constant:
    \begin{align*}
        K &= \prod_{r=1}^2 \prod_{s=1}^2 [1 + \exp(\xi_{rs})].
    \end{align*}

    The relevant conditional probabilities are:
    \begin{align*}
        P(Z_\sigma(m;m')=+1 | X_\sigma, \mathbf{F})
        &= \Lambda(\xi_{11}) (1-\Lambda(\xi_{12})) (1-\Lambda(\xi_{21})) \Lambda(\xi_{22}) \\
        &= \frac{\exp(\xi_{11}) \times 1 \times 1 \times \exp(\xi_{22})}{K} = \frac{\exp(\xi_{11} + \xi_{22})}{K},
        \\
        P(Z_\sigma(m;m')=-1 | X_\sigma, \mathbf{F})
        &= (1-\Lambda(\xi_{11})) \Lambda(\xi_{12}) \Lambda(\xi_{21}) (1-\Lambda(\xi_{22})) \\
        &= \frac{1 \times \exp(\xi_{12}) \times \exp(\xi_{21}) \times 1}{K} = \frac{\exp(\xi_{12} + \xi_{21})}{K}.
    \end{align*}
    It follows that the odds ratio is:
    \begin{align*}
        \frac{P(Z_\sigma(m;m')=+1| X_\sigma, \mathbf{F})}{P(Z_\sigma(m;m')=-1| X_\sigma, \mathbf{F})}
        &= \frac{\exp(\xi_{11} + \xi_{22})}{\exp(\xi_{12} + \xi_{21})} \\
        &= \exp(\xi_{11} + \xi_{22} - \xi_{12} - \xi_{21}) \\
        &= \exp\left( (X_{i_1,j_1}^\prime \beta_0 + \gamma_{j_1} - \lambda_{i_1, m}) + (X_{i_2,j_2}^\prime \beta_0 + \gamma_{j_2} - \lambda_{i_2, m'}) \right. \\
        & \qquad \qquad \left. - (X_{i_1,j_2}^\prime \beta_0 + \gamma_{j_2} - \lambda_{i_1, m}) - (X_{i_2,j_1}^\prime \beta_0 + \gamma_{j_1} - \lambda_{i_2, m'}) \right) \\
        &= \exp\left( (X_{i_1,j_1}^\prime - X_{i_1,j_2}^\prime - X_{i_2,j_1}^\prime + X_{i_2,j_2}^\prime) \beta_0 \right. \\
        & \qquad \qquad \left. + (\gamma_{j_1} - \gamma_{j_1}) + (\gamma_{j_2} - \gamma_{j_2}) - (\lambda_{i_1, m} - \lambda_{i_1, m}) - (\lambda_{i_2, m'} - \lambda_{i_2, m'}) \right) \\
        &= \exp( r_\sigma^\prime \beta_0 + 0 + 0 - 0 - 0 ) \\
        &= \exp( r_\sigma^\prime \beta_0 ).
    \end{align*}
    All fixed effects cancel out due Assumption \ref{a:likelihood_III} and the choice of cutoffs $\mathbf{m}=(m,m,m',m')$.

    It follows that
    \begin{align*}
        P(Z_\sigma(m;m')=1| Z_\sigma(m;m') \in\{-1,1\}, X_\sigma, \mathbf{F})
        &= \frac{\exp\left( r_\sigma^\prime \beta_0 \right)}
             {\exp\left( r_\sigma^\prime \beta_0 \right) + 1} \\
        &= \Lambda\left( r_\sigma^\prime \beta_0 \right).
    \end{align*}
\end{proof}

\subsection{Consistency}

\begin{proof}[Proof of Theorem \ref{thm:consistency_II}]
\label{proof:consistency_II}

We prove parts (a) and (b) separately.
Part (a) shows consistency of ETLE $\check\beta_N$. 
Part (b) shows consistency of PTLE $\widehat\beta_N$.

\paragraph{Part (a): ETLE.}
This proof modifies the proof of Theorem 1 in \textcite{jochmansSemiparametricAnalysisNetwork2018}.
The sample objective function \eqref{eq:sample_objective_function_a} is a sum of $M$ terms, $Q_{mN}(\beta) = \left(q_{mN}^*\right)^{-1} L_{m N}(\beta).$
For each $m$, consider
\[
Q_{mN0}(\beta) = \left(q_{N} p_{mN}\right)^{-1} E \left[ L_{m N}(\beta) \right]
\]
and its limiting
\[
Q_{m0}(\beta) = \lim_{N\to\infty} \left(q_{N} p_{mN}\right)^{-1} E \left[ L_{m N}(\beta) \right].
\]
Each $Q_{m0}$ is concave and uniquely maximized at $\beta_0$ by Assumption \ref{a:identification_II_strong}(2), so the limit $\sum_{m=1}^M Q_{m0}$ of the sample objective function is concave and uniquely maximized at $\beta_0$.

For each $m$, we show mean square convergence of $Q_{mN}$ to $Q_{mN0}$.
Following \textcite{jochmansSemiparametricAnalysisNetwork2018}, we decompose:
\begin{equation*}
    \frac{\sum_\sigma l_{m\sigma}(\beta)}{q^*_{mN}} - \frac{\sum_\sigma E(l_{m\sigma}(\beta))}{q_N p_{mN}}
    =
    \underbrace{\frac{\sum_{\sigma} (l_{m\sigma}(\beta)-E(l_{m\sigma}(\beta)))}{q_N p_{mN}}}_{A_{N1}}
    +
    \underbrace{\frac{\sum_{\sigma} l_{m\sigma}(\beta)}{q_N p_{mN}} \left(\frac{q_N p_{mN}}{q^*_{mN}} - 1\right)}_{A_{N2}}.
\end{equation*}

The second moment of $l_{m\sigma}(\beta)$ is bounded because
\begin{align*}
    \text{Var}(l_{m\sigma}(\beta)) & = E[((l_{m\sigma}(\beta)-E[l_{m\sigma}(\beta)])^2] \\
    & = E(l^2_{m\sigma}(\beta)) - [E(l_{m\sigma}(\beta))]^2 \\
    & \leq E(|l_{m\sigma}(\beta)|^2) \\
    & \leq E[(\log2 + 2 \| r_{\sigma}\| \|\beta\|)^2]\\
    & = (\log 2)^2 + 4\log 2E(\| r_{\sigma}\|)\|\beta\|+4E(\|r_{\sigma}\|^2)\|\beta\|^2,
\end{align*}
where the last inequality uses Assumptions \ref{a:parameter_II} and \ref{a:moments_II}.

For the second moment of the numerator of $A_{N1}$, define $\widetilde l_{m\sigma} = l_{m\sigma}(\beta) - E[l_{m\sigma}(\beta)]$ and
$$D(\sigma) = \left\{ \sigma' : \sigma \text{ and } \sigma' \text{ have at least one node common}\right\}.$$
Then:
\begin{align*}
    E\left[\left(\sum_\sigma \widetilde l_{m\sigma}\right)^2\right]
    &=
    \sum_\sigma \sum_{\sigma'} E \left[\widetilde l_{m\sigma} \widetilde l_{m\sigma'}\right] \\
    &=
    \sum_\sigma \sum_{\sigma' \in D(\sigma)} E \left[\widetilde l_{m\sigma} \widetilde l_{m\sigma'}\right] \\
    &\leq
    \sum_\sigma \sum_{\sigma' \in D(\sigma)} \sqrt{E \left[\widetilde l_{m\sigma}^2 \right] E \left[\widetilde l_{m\sigma'}^2 \right]} \\
    &\leq
    \sum_\sigma \sum_{\sigma' \in D(\sigma)} \frac{E \left[\widetilde l_{m\sigma}^2  \right] + E \left[\widetilde l_{m\sigma'}^2\right]}{2} \\
    &\propto N^3 \sum_\sigma E \left[\widetilde l_{m\sigma}^2  \right] \\
    &= O(N^3 q_{N} p_{mN}),
\end{align*}
where the first line expands the square, the second uses independence of non-overlapping tetrads (Assumption \ref{a:sampling_II}), the third applies Cauchy-Schwarz, the fourth uses the arithmetic-geometric mean inequality, and the fifth uses $|D(\sigma)| \propto N^3$. This dependence structure leads to $E[(\sum_\sigma \widetilde l_{m\sigma})^2] = O(N^3 q_N p_{mN})$, consistent with the variance calculation for sums over dependent quadruples in \textcite[Proof of Theorem 1]{jochmansSemiparametricAnalysisNetwork2018}.

Since $q_N = O(N^4)$, we have $E[A_{N1}^2] = O\left(\frac{N^3 q_N p_{mN}}{(q_N p_{mN})^2}\right) = O\left(\frac{1}{Np_{mN}}\right)$.
Mean square convergence (and thus convergence in probability) of $A_{N1}$ to zero follows from Assumption \ref{a:identification_II_strong}(i), which requires $N p_{mN} \rightarrow \infty$.

For $A_{N2}$, note that $q_{mN}^*/q_N$ is the sample average corresponding to $p_{mN} = E[q^*_{mN}]/q_N$. As established in the proof of Theorem 1 in \textcite{jochmansSemiparametricAnalysisNetwork2018}, the network dependence implies $Var(q^*_{mN}/q_N) = O(1/(Np_{mN}))$. Thus, $(q_{mN}^*/q_N - p_{mN}) \to_p 0$ provided $Np_{mN} \to \infty$ (Assumption \ref{a:identification_II_strong}(i)), which ensures $q^*_{mN}/(q_N p_{mN}) \to_p 1$. The second moment of the first factor $\frac{\sum_{\sigma} l_{m\sigma}(\beta)}{q_N p_{mN}}$ is bounded. Therefore, $A_{N2} \to_p 0$.

Since both $A_{N1}$ and $A_{N2}$ converge to zero in probability, $Q_{mN}(\beta)$ converges pointwise in probability to $Q_{m0}(\beta)$. Uniform convergence follows from standard arguments for M-estimators under compactness (Assumption \ref{a:parameter_II}) and concavity. Since $\sum_m Q_{m0}(\beta)$ is uniquely maximized at $\beta_0$, consistency of $\check\beta_N$ follows.

\paragraph{Part (b): PTLE.}
The scaled sample objective function is
\[
Q_{N}(\beta) = \frac{\sum_m L_{m N}(\beta)}{\sum_m q_{mN}^*} = \frac{\sum_m \sum_\sigma l_{m \sigma}(\beta)}{\sum_m q_{mN}^*}.
\]

Define
\[
Q_{N0}(\beta) = \frac{\sum_m E[L_{m N}(\beta)]}{E[\sum_m q_{mN}^*]} = \frac{\sum_m \sum_\sigma E[l_{m \sigma}(\beta)]}{q_N p_{N}}
\]
where $p_N = \sum_{m=1}^M p_{mN} = E[\sum_m q_{mN}^*] / q_N$. Let $Q_{0}(\beta) = \lim_{N\to\infty} Q_{N0}(\beta)$.
Using demeaned likelihood contributions $\widetilde l_{m \sigma}(\beta) = l_{m \sigma}(\beta) - E[l_{m \sigma}(\beta)]$, decompose:
\begin{align*}
Q_{N}(\beta) - Q_{N0}(\beta) &=
\frac{\sum_{m,\sigma} \widetilde l_{m\sigma}(\beta)}{\sum_m q^*_{mN}}
+
\frac{\sum_{m,\sigma} E[l_{m \sigma}(\beta)]}{\sum_m q^*_{mN}}
-
\frac{\sum_{m,\sigma} E[l_{m \sigma}(\beta)]}{q_N p_{N}} \\
&=
\underbrace{
    \frac{\sum_{m,\sigma} \widetilde l_{m\sigma}(\beta)}{q_N p_{N}}
    \left( \frac{q_N p_N}{\sum_m q^*_{mN}} \right)
}_{A_{N1}'}
+
\underbrace{
     \frac{\sum_{m,\sigma} E[l_{m \sigma}(\beta)]}{ q_N p_{N}}
     \left(
        \frac{ q_N p_{N}}
             {\sum_{m=1}^M q^*_{mN}} - 1
     \right)
}_{A_{N2}'}.
\end{align*}

For the second moment of the numerator of the first term in $A_{N1}'$:
\begin{align*}
E\left[\left(\sum_{m,\sigma} \widetilde l_{m\sigma}\right)^2\right]
&=
\sum_{m,m'} \sum_{\sigma} \sum_{\sigma' \in D(\sigma)} E\left[\widetilde l_{\sigma,m} \widetilde l_{\sigma',m'}\right] \\
&\leq
\sum_{m,m'} \sum_{\sigma} \sum_{\sigma' \in D(\sigma)} \sqrt{E\left[\widetilde l^2_{\sigma,m}\right] E\left[\widetilde l^2_{\sigma',m'}\right]} \\
&\leq
\sum_{m,m'} \sum_{\sigma} \sum_{\sigma' \in D(\sigma)} \frac{E\left[\widetilde l^2_{\sigma,m}\right] + E\left[\widetilde l^2_{\sigma',m'}\right]}{2} \\
&\propto
N^3 M \sum_m \sum_{\sigma} E\left[\widetilde l^2_{\sigma,m}\right] \\
&=
 O(N^3 q_N p_N).
\end{align*}
Most steps are analogous to those in part (a) above.
Note the additional $M$ factor in the fourth line due to the sum over $m'$.
Since $M<\infty$, this term is absorbed in the $O(\cdot)$ notation. 
The final step uses $\sum_m p_{mN} = p_N$ and boundedness of $E[\widetilde l^2_{\sigma,m}]$.

Since $q_N = O(N^4)$, comparing orders shows the second moment $E[(\frac{\sum_{m,\sigma} \widetilde l_{m\sigma}(\beta)}{q_N p_{N}})^2] = O\left(\frac{N^3 q_N p_N}{(q_N p_N)^2}\right) = O\left(\frac{1}{Np_{N}}\right)$.
Mean square convergence follows from Assumption \ref{a:identification_II_weak}(i), which requires $N p_N \rightarrow \infty$.

Now, notice that $\frac{q_N p_N}{\sum_m q^*_{mN}} \to_p 1$ because $\frac{\sum_m q^*_{mN}}{q_N}$ is the sample average corresponding to $p_N = \sum_m p_{mN} = E[\sum_m q^*_{mN}]/q_N$, and its convergence in probability is established by arguments similar to part (a).
This implies $A_{N1}' \to_p 0$, and, with $\frac{\sum_{m,\sigma} E[l_{m \sigma}(\beta)]}{ q_N p_{N}} \to Q_0(\beta)$, it also follows that $A_{N2}' \to_p 0$.

Therefore, $Q_N(\beta) - Q_{N0}(\beta) \to_p 0$. Since $Q_{N0}(\beta) \to Q_0(\beta)$, we have $Q_N(\beta) \to_p Q_0(\beta)$ pointwise. Uniform convergence follows from similar arguments as in Part (a). 
\end{proof}

\printbibliography

\end{document}